\documentclass[journal,two column]{IEEEtran}

\usepackage{amsmath,amssymb}
\usepackage{amsfonts}
\usepackage{graphicx,graphics,color,psfrag}
\usepackage{cite,balance}
\usepackage{caption}
\allowdisplaybreaks
\usepackage{algorithm}
\usepackage{algorithmic}
\usepackage{accents}
\usepackage{amsthm}
\usepackage{bm}
\usepackage{url}
\usepackage[english]{babel}
\usepackage{multirow}
\usepackage{enumerate}
\usepackage{cases}
\usepackage{stfloats}
\usepackage{dsfont}
\usepackage{color,soul}
\usepackage{amsfonts}
\usepackage{cite,graphicx,amssymb}
\usepackage{fancyhdr}
\usepackage{hhline}
\usepackage{array,color}
\usepackage{mathtools}
\usepackage{subcaption}
\graphicspath{{./Figs/}}
\newtheorem{definition}{\emph{\underline{Definition}}}

\newtheorem{lemma}{\emph{\underline{Lemma}}}

\newtheorem{proposition}{\emph{\underline{Proposition}}}

\newtheorem{example}{\bf \emph{\underline{Example}}}
\newtheorem{remark}{\bf \emph{\underline{Remark}}}

\def\({\left(}
\def\){\right)}

\def\b0{{\mathbf{0}}}

\makeatletter
\g@addto@macro\normalsize{%
	\setlength\abovedisplayskip{6pt}  
	\setlength\belowdisplayskip{6pt}  
}
\makeatother

\title{Physical-Layer Security in Mixed Near-Field and Far-Field Communication Systems}

\author{Tianyu Liu, Changsheng You,~\IEEEmembership{Member,~IEEE}, Cong Zhou, Yunpu Zhang, Shiqi Gong, \\ Heng Liu, Guangchi Zhang
	\vspace{-20pt}
	\thanks{Tianyu Liu and Changsheng You are with the Department of Electronic and Electrical Engineering, Southern University of Science and Technology, Shenzhen 518055, China (e-mail: $\!$liuty2022@mail.sustech.edu.cn; youcs@sustech.edu.cn).}
	\thanks{Cong Zhou is with the School of Electronics and Information Engineering, Harbin Institute of Technology, Harbin, 150001, China (e-mail: zhoucong@stu.hit.edu.cn).}
	\thanks{Yunpu Zhang is with the Department of Electrical Engineering, City University of Hong Kong, Hong Kong (e-mail: yunpu.zhang@my.cityu.edu.hk).}
	\thanks{Shiqi Gong and Heng Liu are with the School of Information and Electronics, Beijing Institute of Technology, Beijing 100081, China (e-mail: gsqyx@163.com; heng$\_$liu$\_$bit$\_$ee@163.com).}
	\thanks{Guangchi Zhang is with the School of Information Engineering,
	Guangdong University of Technology, Guangzhou 510006, China (e-mail: gczhang@gdut.edu.cn).
	}
	\thanks{\emph{Corresponding author: Changsheng You.}
	}
	}

\begin{document}

\captionsetup[figure]{name={Fig.},labelsep=period,singlelinecheck=off} 

\maketitle

\begin{abstract}
\emph{Extremely large-scale arrays} (XL-arrays) have emerged as a promising technology to improve the spectrum efficiency and spatial resolution of future wireless systems. 
Different from existing works that mostly considered physical layer security (PLS) in either the far-field or near-field, we consider in this paper a new and practical scenario, where legitimate users (Bobs) are located in the far-field of  a base station (BS) while eavesdroppers (Eves) are located in the near-field for intercepting confidential information at short distance, referred to as the \emph{mixed near-field and far-field PLS}.
Specifically, we formulate an optimization problem to maximize the sum-secrecy-rate of all Bobs by optimizing the power allocation of the BS, subject to the constraint on the total BS transmit power. 
To shed useful insights, we first consider a one-Bob-one-Eve system and characterize the insecure-transmission region of the Bob in closed form. 
Interestingly, we show that the insecure-transmission region is significantly \emph{expanded} as compared to that in conventional far-field PLS systems, due to the \emph{energy-spread effect} in the mixed-field scenario.
Then, we further extend the analysis to a two-Bob-one-Eve system. It is revealed that as compared to the one-Bob system, the interferences from the other Bob can be effectively used to weaken the capability of Eve for intercepting signals of target Bobs, thus leading to enhanced secrecy rates.
Furthermore, we propose an efficient algorithm to obtain a high-quality solution to the formulated non-convex problem by leveraging the successive convex approximation (SCA) technique. 
Finally, numerical results demonstrate that our proposed algorithm achieves a higher sum-secrecy-rate than the benchmark scheme where the power allocation is designed based on the (simplified) far-field channel model.
 
\end{abstract}

\vspace{-10pt}
\section{Introduction}
Extremely large-scale arrays (XL-arrays) have emerged as a promising technology for future sixth-generation (6G) wireless systems to achieve ultra-high spectral efficiency and spatial resolution, thereby supporting emerging applications such as metaverse and digital twin \cite{Tian_Jiachen2023, Cui_Mingyao2023v10}.  With drastically growing number of antennas in high-frequency bands, the well-known \emph{Rayleigh distance} will be greatly expanded. 
This thus induces a fundamental paradigm shift in electromagnetic (EM) field characteristics, transitioning from conventional far-field communications to the emerging \emph{near-field communications} \cite{Liu_Yuanwei2025}. 

In particular, compared with far-field communication systems featured by planar wavefronts (PWs),  a more accurate spherical wavefronts (SWs) model needs to be considered in the channel modelling of near-field systems.
This new channel characteristic gives rise to an appealing feature,
called near-field \emph{beam-focusing}, where the beam energy can be concentrated at a specific spatial location rather than along a spatial angle as in far-field systems\cite{Zhang_Haiyang2022}.
This thus stimulates upsurging research on new near-field beam training techniques to achieve two-dimensional beam search over both the angle and distance domains \cite{Cui_Mingyao2022,Lu_Yu2024,Liu_Linyangside}.
In addition, the near-field energy-focusing effect provides a new degree-of-freedom (DoF) to mitigate inter-user interference in both the angle and distance domains, thereby achieving  the enhanced system performance \cite{Zhang_Haiyang2022}. 
While the existing works have mainly considered either the near-field or far-field communications, they overlooked a practical mixed near-field and far-field communications scenario where there may exist both near-field and far-field users. To characterize the difference between near-field and far-field regions from the perspective of array gains, a new metric called effective Rayleigh distance, was introduced in \cite{Cui_Mingyao2024}. For example, consider a wireless system where the base station (BS) operating at $f=30$ GHz is equipped with an uniform linear array (ULA) of 256 antennas. The corresponding effective Rayleigh distance is about 120 meters (m). Therefore, in practical cellular networks (e.g., with a cell radius of 200 m), users may reside in either the near-field or far-field region. For the mixed-field system, it was shown in \cite{YZhang2023} that when the XL-array steers information beams towards both the near-field and far-field users, near-field users may suffer strong interference from far-field oriented beams, even when the users are located at different angles. Furthermore, such power leakage from the far-field beam was leveraged in \cite{Zhang_Yunpu2024} to efficiently charge near-field energy harvesting receivers. 

In addition, physical layer security (PLS) is a promising approach to protect wireless communication. The existing works on PLS mainly considered the near-field or far-field systems. Specifically, for far-field PLS systems, various beamforming designs have been studied in the angular domain to enhance achievable system secrecy rates. One strategy is to design precoding at the transmitter to mitigate information leakage at Eves.
To this end, zero-forcing (ZF) precoding was applied in \cite{Li_Yiqing2017} to transmit confidential signals within the null space of eavesdroppers (Eves), such that the Eve can not overhear information intended to legitimate users (Bobs). 
Another strategy involves transmitting artificial jamming signals via jamming-aided beamforming to interfere with reception of Eves. 
Artificial noise (AN), as an effective jamming technique, was introduced in \cite{Zhao_Nan2019} to interfere with signal reception of Eves by precisely aligning AN-aided beams with the location of Eves, under the assumption of perfect channel state information (CSI) of Eves.
Furthermore, when the CSI of Eves is unavailable, AN was designed to be uniformly distributed on the null space of the channels of Bobs\cite{Zhou_Fuhui2018}. As such, only the information reception of Eves was degraded by the AN, while that of Bobs remained nearly unaffected.
However, when Bob and Eve are located at the same spatial angle, the aforementioned strategies may suffer from performance degradation due to the highly correlated channels between Bobs and Eves. To address this issue, intelligent reflecting surface (IRS) was employed in \cite{Cui_Miao2019} to enable secure transmission\cite{You_Changsheng2020, Shao_Xiaodan2022}. Specifically, by reconfiguring the wireless propagation environment through tunable reflecting elements to reduce received signal power at Eves, IRSs significantly enhanced the secrecy performance.
Besides, movable antennas (MAs), as a form of reconfigurable antennas, can also reconfigure the wireless channel by flexibly adjusting the positions of transmit/receive antennas within a specified region. This technique was employed in \cite{Hu_Guojie2024} to maximize the achievable secrecy rate by jointly optimizing the transmit beamforming and the positions of all antennas, leading to significant performance improvement thank to the additional spatial DoF.

For near-field PLS systems, interestingly, it was shown in \cite{Zhang_Zheng2024} that near-field beam focusing can be leveraged to protect data transmission security of Bobs, even when Eves were located at shorter distance than Bobs and/or they were at the same angle.
Specifically, it was revealed in \cite{Boqun_Zhao2024} that, in the near-field, Bob can achieve a non-zero secrecy-capacity even when Bob and Eve were located at the same angle owing to the energy-focusing effect. In contrast, the secrecy-capacity dropped to zero in the far-field under the same angular condition.
Moreover, a new concept of depth-of-insecurity was proposed in \cite{Boqun_Zhao2024} to evaluate secrecy performance, which was shown to decrease with an increasing number of transmit antennas. 
In addition, AN-based near-field PLS system was further studied in \cite{Yunpu_Zhang2024_secure}, where AN-aided beams were steered towards Eves to effectively interfere with information reception of Eves, while causing negligible impact on information reception of Bobs due to the energy-focusing effect. This resulted in a performance rate gain compared to the far-field counterpart, in which Bobs also suffered considerable interference from AN, especially when Bob and Eve were located at the same angle. 
Furthermore, a near-field wideband PLS communication system was studied in \cite{Zhang_Yuchen2024}, where the beam-split effect inherent in wideband transmissions leads to severe information leakage. To address this problem, a true-time delayers (TTDs)-based analog architecture was designed to eliminate the beam-split effect, thereby greatly enhancing the secrecy performance.

\subsection{Motivations and Contributions}
However, the existing results on far-field or near-field PSL systems may not be applicable to the practical mixed near-field and far-field systems, due to the intricate correlation between the near-field and far-field channels \cite{Cui_Mingyao2022_3}. 
Specifically, for one-Bob systems, when a BS sends information to a far-field Bob via angle-based beamforming (e.g., discrete Fourier transform (DFT) beam), it may cause information leakage to near-field Eves due to the energy-spread effect \cite{YZhang2023}, even when they are not located at the same angle. 
Next, for multi-Bob systems, far-field Bobs may introduce severe interference to the near-field Eve due to the energy-spread effect, hence potentially mitigating the eavesdropping ability of Eve on the target Bobs. 
These thus motivate us to study a new mixed-field PLS system for secrecy-rate maximization. 
In particular, we aim to answer the following three questions:
\begin{itemize}
\item First, what is the insecure-transmission region in the mixed-field PLS system, i.e., the region where the secrecy rate of Bobs drop to zero?
\item How to design the power allocation to maximize the sum-secrecy-rate in the mixed-field PLS system? 
\item Are there any gains or losses in the mixed-field secrecy-rate performance as compared to the far-field case?
\end{itemize}

To this end, we consider in this paper a new and challenging mixed-field PLS system shown in Fig. \ref{system_setting}, where a BS equipped with an XL-array transmits confidential messages to multiple far-field Bobs in the presence of multiple non-cooperative Eves located in its near-field region. 
We formulate an optimization problem to maximize the sum-secrecy-rate by optimizing the power allocation of the BS for multiple Bobs.
To shed useful insights into the new characteristics of mixed-field PLS system, we first consider a one-Bob-one-Eve system and characterize the \emph{insecure-transmission region} of the Bob in closed form.
Notably, we analytically show that the insecure-transmission region of Bob in the mixed-field is much \emph{larger} than that in far-field systems.
Then, we further extend the theoretical analysis to a two-Bob-one-Eve system. 
It is revealed that as compared to the one-Bob system, the interferences from other Bobs can be effectively used to mitigate the capability of Eve to intercept signals of target Bobs, leading to enhanced secrecy rate performance. 
Subsequently, an efficient algorithm is proposed to solve the formulated non-convex optimization problem for the general case of multiple Eves and Bobs, by leveraging the efficient successive convex approximation (SCA) technique. 
Finally, numerical results demonstrate that our proposed algorithm achieves a higher sum-secrecy-rate than the benchmark scheme where the power allocation is designed based on the (simplifed) far-field channel model. 
This highlights the importance of considering practical mixed-field channel models for PLS system designs.

\vspace{-10pt}
\subsection{Organization and Notations}
The remainder of this paper is organized as follows. Section II presents the system model and problem formulation for mixed-field PLS systems. In Section III,  we consider two special cases to gain useful insights into the new characteristics of the mixed-field PLS system. In Section IV, we propose an efficient algorithm to obtain a suboptimal solution for the general case. Numerical results are presented in Section V to evaluate the performance of the proposed scheme, followed by conclusions given in Section VI.

\emph{Notations:}  Lowercase and uppercase boldface letters are
used to represent vectors and matrices, such as $\mathbf{a}$ and $\mathbf{A}$. For vectors and matrices, the superscripts  $(\cdot)^H$ and $(\cdot)^T$ stand for the transpose and Hermitian transpose, respectively.
Additionally, calligraphic letters (e.g., $\mathcal{A}$) denote discrete and finite sets. The notations $|\cdot|$,  $\|\cdot\|$, and $\|\cdot\|_F$ refer to the absolute value of a scalar, the $\ell_2$ norm of  a vector, and the Frobenius norm of a matrix, respectively. The symbol $\mathbf{I}_K$ represents a $K$-dimensional identity matrix. $\mathcal{O}(\cdot)$ denotes the standard big-O notation.

\begin{figure}[t]
\centering
\includegraphics[width=3.5in]{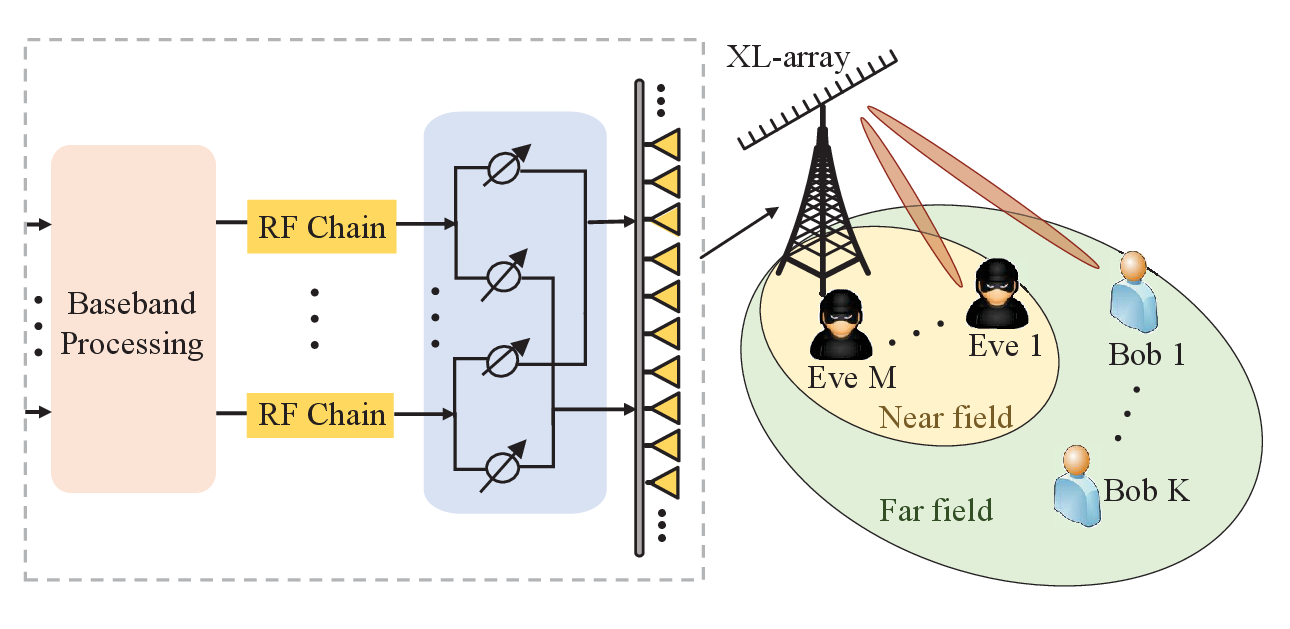}
\caption{The considered mixed-field PLS system.}
\label{system_setting}
\vspace{-10pt}
\end{figure}
\section{System Model and Problem Formulation}
\subsection{System Model}
We consider a downlink PLS communication system as shown in Fig. \ref{system_setting}, where a BS equipped with an $N$-antenna XL-array transmits confidential messages to $K$ single-antenna legitimate users (Bobs) in the presence of $M$ single-antenna eavesdroppers (Eves).  We assume that the XL-array BS is placed at the $y$-axis with its center at the origin $(0,0)$. As such, the coordinate of the $n$-th antenna is given by $(0, \delta_n d)$ with $\delta_n=\frac{2n-N-1}{2}, n \in \mathcal{N}=\{0,1,\ldots,N-1\}$. The XL-array has an aperture size of $D = (N-1)d$, where $d=\lambda/2$ denotes the half-wavelength inter-antenna spacing with $\lambda$ denoting the carrier wavelength. 

For the boundary between the near-field and far-field regions of the XL-array, we consider the \emph{effective} Rayleigh distance in this paper, which is defined as $Z(\theta)=\epsilon\frac{2D^2(1-\theta^2)}{\lambda}$ with $\theta=\sin(\phi)$ representing the \emph{spatial} angle of the user, $\phi$ being the physical angle, and $\epsilon=0.367$ \cite{Cui_Mingyao2024}. It is worth noting that compared with conventional Rayleigh distance defined based on phase variations, the effective Rayleigh distance specifies the distance boundary where the \emph{array gains} obtained from far-field and near-field channel models exhibit notable differences, thereby directly affecting the communication rate performance (such as secrecy rate in this paper)\cite{Cui_Mingyao2024}.
For the locations of users, we consider a challenging scenario where Bobs, denoted by $\mathcal{K}=\{1,2,\dots,K\}$, are located in the far-field region of the BS, while Eves, denoted by $\mathcal{M}=\{1,2,\dots,M\}$, are assumed to be located in the near-field region for intercepting confidential information at a short distance, while the results in this paper can be extended to other scenarios (see discussions in \textbf{Remark \ref{othercase}}). 

\subsubsection{Channel models} The channel models for the near-field Eves and far-field Bobs are introduced as follows.

\underline{\textbf{Near-field channel model for Eves:}} For the near-field Eve $m \in \mathcal{M}$, its channel from the BS can be modeled as 
\begin{align}
    \mathbf{h}^H_{{\rm E},m} &= (\mathbf{h}^{\rm LoS}_{{\rm E},m})^H+\sum_{\ell_m=1}^{L_m}(\mathbf{h}^{\rm NLoS}_{{\rm E},\ell_m})^H,
\end{align}
which consists of one line-of-sight (LoS) path and $L_m$ non-LoS (NLoS) paths. Specifically, the LoS component of the near-field channel from the BS to Eve $m$ can be modeled as 
\begin{align}
&(\mathbf{h}^{\rm LoS}_{{\rm E},m})^H = \sqrt{N}h_{{\rm E},m}\mathbf{b}^H(\theta_{{\rm E},m},r_{{\rm E},m}),
\end{align} 
where $h_{{\rm E},m}=\frac{\sqrt{\beta}}{ r_{{\rm E},m}}e^{-\frac{j2\pi r_{{\rm E},m}}{\lambda}}$ is the complex-valued channel gain, with $\beta$ and $r_{{\rm E},m}$ representing the reference channel gain at a distance of 1 m and the distance between the BS array center and Eve $m$, respectively. $\mathbf{b}^H(\theta_{{\rm E},m},r_{{\rm E},m})$ denotes the near-field channel steering vector from the BS to Eve $m$, which is modeled based on uniform SWs (USWs)  \cite{Wu_Xun2024}
\begin{align}
	\vspace{-3pt}
    &\mathbf{b}^H\left(\theta_{{\rm E},m},r_{{\rm E},m}\right)\nonumber\\
    &=\!\frac{1}{\sqrt{N}}\!\left[\!e^{-j2\pi(\!r^{(0)}_{{\rm E},m}-r_{{\rm E},m})/\lambda},\dots,e^{-j2\pi(r^{(N-1)}_{{\rm E},m}-r_{{\rm E},m})/\lambda}\!\right]\!.
    \vspace{-3pt}
\end{align}
Herein, $\theta_{{\rm E},m} \triangleq \sin{\phi_{{\rm E},m}}$ is the spatial angle at the BS with $\phi_{{\rm E},m} \in (-\frac{\pi}{2},\frac{\pi}{2})$ denoting the physical angle-of-departure (AoD) from the BS array center to Eve $m$; and $r^{(n)}_{{\rm E},m}=\sqrt{r^2_{{\rm E},m}+\delta^2_nd^2-2r_{{\rm E},m}\theta_{{\rm E},m} \delta_n d}$ represents the distance between the $n$-th antenna of the BS and Eve $m$.

Similarly, the NLoS components can be modeled as
\begin{align}
&(\mathbf{h}^{\rm NLoS}_{{\rm E},\ell_m})^H = \sqrt{\frac{N}{L_m}}h_{{\rm E},\ell_m}\mathbf{b}^H(\theta_{{\rm E},\ell_m},r_{{\rm E},\ell_m}),
\end{align}
where $h_{{\rm E},\ell_m}$ is the corresponding complex-valued channel gain of the $\ell_m$-th NLoS path; and $r_{{\rm E},\ell_m}$ and $\theta_{{\rm E},\ell_m}\triangleq \cos{\phi_{{\rm E},\ell_m}}$ denote the distance and spatial angle of the $\ell_m$-scatterer with respect to (w.r.t.) the origin with $\phi_{{\rm E},\ell_m}$ being the physical AoD.

\underline{\textbf{Far-field channel model for Bobs:}} For the far-field Bob $k$, its channel from the BS can be modeled based on  uniform PWs (UPWs), which is given by 
\begin{align}
        \mathbf{h}^{H}_{{\rm B},k}\!=\!\sqrt{N}h_{{\rm B},k}\mathbf{a}^H(\theta_{{\rm B},k})\!+\!\sqrt{\frac{N}{L_k}}\!\sum_{\ell_k=1}^{L_k}\!h_{{\rm B},\ell_k}\mathbf{a}^H(\theta_{\rm B,\ell_k}),
\end{align}
where $h_{{\rm B},k} = \frac{\sqrt{\beta}}{ r_{{\rm B},k}}e^{-\frac{j2\pi r_{{\rm B},k}}{\lambda}}$ represents the complex-valued channel gain with $r_{{\rm B},k}$ denoting the distance between the array center and Bob $k$. $\mathbf{a}^H(\theta_{{\rm B},k})$ denotes the far-field channel steering vector from the XL-array BS to Bob $k$, which is modeled as
\begin{align}
    \mathbf{a}^H(\theta_{{\rm B},k}) = \frac{1}{\sqrt{N}}\left[1, e^{j \pi \theta_{{\rm B},k}}, \dots, e^{j \pi (N-1) \theta_{{\rm B},k}} \right],
\end{align}
where $\theta_{{\rm B},k} \triangleq \sin{\phi_{{\rm B},k}}$ represents the spatial angle at the BS with $\phi_{{\rm B},k} \in (-\frac{\pi}{2},\frac{\pi}{2})$ denoting the AoD from the array center to Bob $k$.

\subsubsection{Signal model}
Let $x_{{\rm B},k}$, $k \in \mathcal{K}$ denote the information-bearing signal sent by the BS to Bob $k$. To reduce the power and hardware cost in high-frequency bands, we consider the hybrid beamforming architecture at the BS to serve Bobs with $K$ ratio frequency (RF) chains \cite{Wang_Zhaolin2025}. As such, the transmitted signal by the BS can be expressed as $\bar{\mathbf{x}}=\mathbf{F}_{\rm A}\mathbf{F}_{\rm D}\mathbf{x}$, where $\mathbf{x}=[x_{{\rm B},1}, x_{{\rm B},2}, \cdots, x_{{\rm B},K}]^T$, $\mathbf{F}_{\rm D}$ represents a $K\times K$ digital precoder and $\mathbf{F}_{\rm A}=[\mathbf{w}_{{\rm B},1},\mathbf{w}_{{\rm B},2},\cdots, \mathbf{w}_{{\rm B},K}]$ denotes an $N\times K$ analog precoder with $\mathbf{w}_{{\rm B},k}$ denoting the analog beamforming vector for Bob $k$. 

We assume that the BS has perfect CSI of all Bobs and Eves, which can be acquired by using existing far-field and near-field channel estimation and beam training methods (e.g., \cite{Zhang_Yunpu2022,Han_Yu2020,You_Changsheng2020,You_Changsheng202038}).\footnote{The results in \cite{Mukherjee2012} revealed that even for a passive Eve, it is still possible to estimate its CSI through the local oscillator power inadvertently leaked from the Eve's receiver RF frontend. In addition, this paper can be extended to the case where only imperfect CSI of Eve is known at the BS, by modeling the channel of Eve based on deterministic model accounting for CSI uncertainty\cite{Xu_Dongfang2020}.}  As such, the received signal at Bob $k$ is given by
\begin{align}
	y_{{\rm B},k} &= \mathbf{h}^H_{{\rm B},k}\mathbf{F}_{\rm A}\mathbf{F}_{\rm D}\mathbf{x}+z_{{\rm B},k},
\end{align}
where $z_{{\rm B},k}$ with power $\sigma^2$ is the received additive white Gaussian noise (AWGN). The corresponding achievable information rate of Bob $k$, in bits per second per Hertz (bps/Hz), is given by
\begin{align}
	&R_{{\rm B},k}=\log _2\left(1+\frac{\left|\mathbf{h}^H_{{\rm B},k}\mathbf{F}_{\rm A}\mathbf{f}_{{\rm D},k}\right|^2}{\sum^K_{i=1,i \neq k}\left|\mathbf{h}^H_{{\rm B},k}\mathbf{F}_{\rm A}\mathbf{f}_{{\rm D},i}\right|^2+\sigma^2}\right), \label{AchinforRate}
\end{align}
where $\mathbf{f}_{{\rm D},k}$ represents the $k$-th column of $\mathbf{F}_{\rm D}$.

Meanwhile, the received signal at Eve $m$ for intercepting information $x_{{\rm B},k}$ is given by
\begin{align}
y_{{\rm E},m,k}=\mathbf{h}^H_{{\rm E},m}\mathbf{F}_{\rm A}\mathbf{F}_{\rm D}\mathbf{x}+z_{{\rm E},m},
\end{align}
where $z_{{\rm E},m}$ with power $\sigma^2$ is the received AWGN. The eavesdropping rate (in bps/Hz) of Eve $m$ for wiretapping $x_{{\rm B},k}$ is thus given by
\begin{align}
	&R_{{\rm E},m,k}=\log _2\left(1\!+\!\frac{ \!\left|\mathbf{h}^H_{{\rm E},m}\mathbf{F}_{\rm A}\mathbf{f}_{{\rm D},k}\right|^2}{\sum^K_{i=1,i \neq k}\!\left|\mathbf{h}^H_{{\rm E},m}\mathbf{F}_{\rm A}\mathbf{f}_{{\rm D},i}\right|^2\!\!+\!\!\sigma^2}\!\right)\!. \label{Acheverate}
\end{align}

Similar to \cite{Anaya_López2023,Li_Qiang2013,Ng_Derrick2013}, we consider that the $M$ Eves operate independently without cooperation in intercepting the signals transmitted to Bobs. As such, the eavesdropping rate for intercepting the data to Bob $k$ (i.e., $x_{{\rm B},k}$) is determined by the maximum eavesdropping rate among all Eves, which is given by $R_{{\rm E},k} =\max \left\{ R_{{\rm E},m,k}, m \in \mathcal{M} \right\}$. As such, the achievable-secrecy-rate (in bps/Hz) of Bob $k$ is given by
\begin{align}
    R_{{\rm S},k} &= \left[R_{{\rm B},k}-R_{{\rm E},k}\right]^{+}\nonumber\\
    &=\left[R_{{\rm B},k}-\max \left\{ R_{{\rm E},m,k}, m \in \mathcal{M} \right\}\right]^{+},\label{SecuRate}
\end{align}
where $[x]^{+}\triangleq\max\{x,0\}$. 

\subsection{Problem Formulation}
In this subsection, an optimization problem is formulated to maximize the sum-secrecy-rate of all Bobs by jointly optimizing the digital beamforming, i.e., $\mathbf{F}_{\rm D}$, and the analog beamforming, i.e., $\mathbf{F}_{\rm A}$ as follows
\begin{align}
	\text{(P1)}\;\; \max_{\mathbf{F}_{\rm A},\mathbf{F}_{\rm D}} \quad &\sum^K_{k=1} R_{{\rm S},k} \label{P1sumrate}\\
	\text{s.t.}\;\;\;\;\;&|[\mathbf{F}_{\rm A}]_{n,k}|=1, \forall n \in \mathcal{N}, k \in \mathcal{K}, \label{ConstantModulus}\\
	&\|\mathbf{F}_{\rm A}\mathbf{F}_{\rm D}\|^2_F \leq P_{\rm tot}, \label{TotPower}
\end{align}
where \eqref{ConstantModulus} is the constant-modulus constraint, and \eqref{TotPower} is the transmit power constraint for the BS with $P_{\rm tot}$ denoting its maximum transmit power. This optimization problem can be solved by using an efficient hybrid beamforming method that minimizes the difference between the optimal digital beamforming and the hybrid beamforming\cite{Wang_Zhaolin2025}, as will be discussed in \textbf{Remark \ref{AlgorithmForhybrid}}.

In this paper, we aim to study the effect of beam scheduling for Bobs on the system PLS performance. Therefore, to obtain useful insights as well as reduce hardware cost of XL-arrays, we consider in the sequel a low-complexity hybrid beamforming scheme \cite{Zhang_Yunpu2024}, where $K$ analog beams are steered towards individual Bobs (i.e., $\mathbf{w}_{{\rm B},k}=\mathbf{a}(\theta_{{\rm B},k})$) and their beam powers ($\mathbf{p}=\{P_{{\rm B},k}, k \in \mathcal{K}\}$) can be flexibly controlled. As such, the digital beamforming can be simplified as an identity matrix, i.e., $\mathbf{F}_{\rm D}=\mathbf{I}_{K}$ \cite{Zhang_Yunpu2024,Jin_Juening2019}\footnote{It is worth noting that this low-complexity scheme can be easily implemented in practice. Moreover, it will be shown in Section V that its performance is close to that of a more complicated scheme discussed in \textbf{Remark \ref{AlgorithmForhybrid}}, where digital beamforming (i.e., $\mathbf{F}_{\rm D}$) is further devised to enhance PLS performance.}. Under the above assumptions, the
achievable-secrecy-rate (in bps/Hz) of Bob $k$ can be rewritten as \eqref{securrateMRT}, which is presented at the bottom of the next page. Then, Problem (P1) reduces to 
\begin{align}
\text{(P2)}\;\; \max_{\mathbf{p}} \sum^K_{k=1} R_{{\rm S},k} 
~~~~~~\text{s.t.} \sum^K_{k=1}P_{{\rm B},k} \leq P_{\rm tot}. \label{P1powercons}
\end{align}

Problem (P2) is a non-convex optimization problem and hence difficult to be optimally solved, since 1) the objective function takes the form of a logarithmic fraction with strongly coupled power allocation variables, and 2) the operator $\left[\cdot\right]^+$ makes the objective function non-smooth and the operator $\max\{\cdot\}$ renders Problem (P2) more difficult to be optimally solved. On the other hand, directly solving Problem (P2) yields few insights into the PLS in mixed-field scenarios. To address these issues, we first consider two special cases with a small number of Bobs and Eves, and characterize the insecure-transmission region of Bobs in closed-form. 
Next, we further propose an efficient algorithm to obtain a high-quality solution to Problem (P1) for the general case.

\begin{figure*}[b]
	\hrulefill 
	\vspace{0.5em} 
	\begin{equation}
		\resizebox{\textwidth}{!}{$
			R_{{\rm S},k}\!=\!\left[\!\log _2\left(1+\frac{N P_{{\rm B},k}|h_{{\rm B},k}|^2}{\sum^K_{i=1,i \neq k}NP_{{\rm B},i}|h_{{\rm B},k}|^2\left|\mathbf{a}^H(\theta_{{\rm B},k})\mathbf{a}^H(\theta_{{\rm B},i})\right|^2+\sigma^2}\right)\!-\!\max\left\{\log _2\left(1+\frac{N P_{{\rm B},k}|h_{{\rm E},m}|^2\left|\mathbf{b}^H(\theta_{{\rm E},m},r_{{\rm E},m})\mathbf{a}^H(\theta_{{\rm B},k})\right|^2}{\sum^K_{i=1,i \neq k}NP_{{\rm B},i}|h_{{\rm E},m}|^2\left|\mathbf{b}^H(\theta_{{\rm E},m},r_{{\rm E},m})\mathbf{a}^H(\theta_{{\rm B},i})\right|^2+\sigma^2}\right), m \in \mathcal{M}\right\} \!\right]^{+} $}.
		\label{securrateMRT} 
	\end{equation}
	\vspace{-1em} 
\end{figure*}

\section{What is the Secure-Transmission Region in the Mixed-Field PLS System ?}
To shed important insights into the new characteristics of mixed-field PLS, we investigate in this section two special cases for Problem (P2), where there are 1) one Bob and one Eve, and 2) two Bobs and one Eve. For the former case, we reveal that the insecure-transmission region of Bob in the mixed-field PLS system is significantly expanded due to the mixed-field channel correlation (or energy-spread effect), as compared to conventional far-field PLS systems. Furthermore, we show that in two-Bob systems, the interference from Bobs to Eve can be strategically utilized to mitigate eavesdropping capabilities of the Eve, thereby enhancing the sum-secrecy-rate.

To facilitate theoretical analysis, we mainly consider the communication systems operating in high-frequency bands, for which the LoS path has much stronger power than the NLoS paths\footnote{The main analytical results in the LoS case are applicable to the general multipath channel case, as will be shown in Fig. \ref{simulation_result1}. Moreover, the performance of proposed algorithm under general multipath channels will be numerically evaluated in Section V.}. As such, the channels from the BS to Eve $m$ and that from the BS to Bob $k$ can be approximated by their LoS components, which are respectively obtained as $\mathbf{h}^H_{{\rm E},m}\approx \sqrt{N}h_{{\rm E},m}\mathbf{b}^H(\theta_{{\rm E},m},r_{{\rm E},m})$ and $\mathbf{h}^H_{{\rm B},k}\approx \sqrt{N}h_{{\rm B},k}\mathbf{a}^H(\theta_{{\rm B},k})$.

\vspace{-4pt}
\subsection{Case I: One-Bob-One-Eve System}
First, we consider the case where there is only one Bob and one Eve in the mixed-field PLS system. As such, the power allocation is not needed.  Based on the above, the objective function in Problem (P1) reduces to the achievable-secrecy-rate below (with users' indices omitted for brevity)
\begin{align}
     R_{\rm S} &\!=\!\left[ R_{\rm B}-R_{\rm E} \right]^+\nonumber\\
     &\!=\!\left[\!\log _2\!\left(1\!+\!\frac{P_{\rm B}\left|\mathbf{h}_{\rm B}^H \mathbf{w}_{\rm B}\right|^2}{\sigma^2}\right)\!-\!\log _2\left(1\!+\!\frac{P_{\rm B} \left|\mathbf{h}_{\rm E}^H \mathbf{w}_{\rm B}\right|^2}{\sigma^2}\right)\!\right]^+\nonumber\\
  &=\!\left[\!\log_2\left(1+\frac{\frac{\mu}{r^2_{\rm B}}\left(\frac{1}{r^2_{\rm B}}-\frac{1}{r^2_{\rm E}}\left|\mathbf{b}^{H}\left(\theta_{\rm E},r_{\rm E}\right)\mathbf{a}\left(\theta_{\rm B}\right)\right|^2\right)}{1+\frac{\mu}{r^2_{\rm E}} \left|\mathbf{b}^{H}\left(\theta_{\rm E},r_{\rm E}\right)\mathbf{a}\left(\theta_{\rm B}\right)\right|^2}\right)\!\right]^+,\label{SecuRate1}
\end{align}
where $\mu \triangleq NP_{\rm B}\beta/\sigma^2$. Next, we first characterize the insecure-transmission region and the achievable-secrecy-rate of the Bob.

It is observed from \eqref{SecuRate1} that, the achievable-secrecy-rate of Bob is jointly determined by the path-loss of Bob and Eve, i.e., $1/r^2_{\rm B}$ and $1/r^2_{\rm E}$, as well as the channel correlation between Bob and Eve, i.e., $\left|\mathbf{b}^{H}\left(\theta_{\rm E},r_{\rm E}\right)\mathbf{a}\left(\theta_{\rm B}\right)\right|^2$. Note that in the far-field PLS system, when Bob and Eve are located at different angles, their channel correlation is usually \emph{small/negligible} due to channel (near-) orthogonality in massive/XL-MIMO systems. However, this conclusion may not hold in the mixed-field scenario due to the intricate correlation between the near-field channel (for Eve) and the far-field beam (steered to Bob), which is formally defined as follows.

\subsubsection{Insecure-transmission region of Bob}
Based on \eqref{SecuRate1}, the data transmissions to Bob are secure when $R_{\rm S}=\left[ R_{\rm B}-R_{\rm E} \right]^+>0$ \cite{Li_Bin2019}, leading to the following result.
\begin{lemma}\label{securitycon_initial}\rm (Secure-transmission condition)
    In the one-Bob-one-Eve system, the data transmissions to Bob are secure when
    \begin{align}
        \left|\mathbf{b}^{H}\left(\theta_{\rm E},r_{\rm E}\right)\mathbf{a}\left(\theta_{\rm B}\right)\right|^2<\frac{r^2_{\rm E}}{r^2_{\rm B}}. \label{SecuCond}
    \end{align}
\end{lemma}
\begin{proof}
According to the secrecy rate in \eqref{SecuRate1}, the data transmissions to Bob are secure when $1/r^2_{\rm B} - \left|\mathbf{b}^{H}\left(\theta_{\rm E},r_{\rm E}\right)\mathbf{a}\left(\theta_{\rm B}\right)\right|^2/r^2_{\rm E} > 0$, which can be rewritten as \eqref{SecuCond}, thereby completing the proof.
\end{proof}
\begin{figure}[H]
	\centering
	\begin{subfigure}[b]{0.49\linewidth}
		\includegraphics[width=1\linewidth]{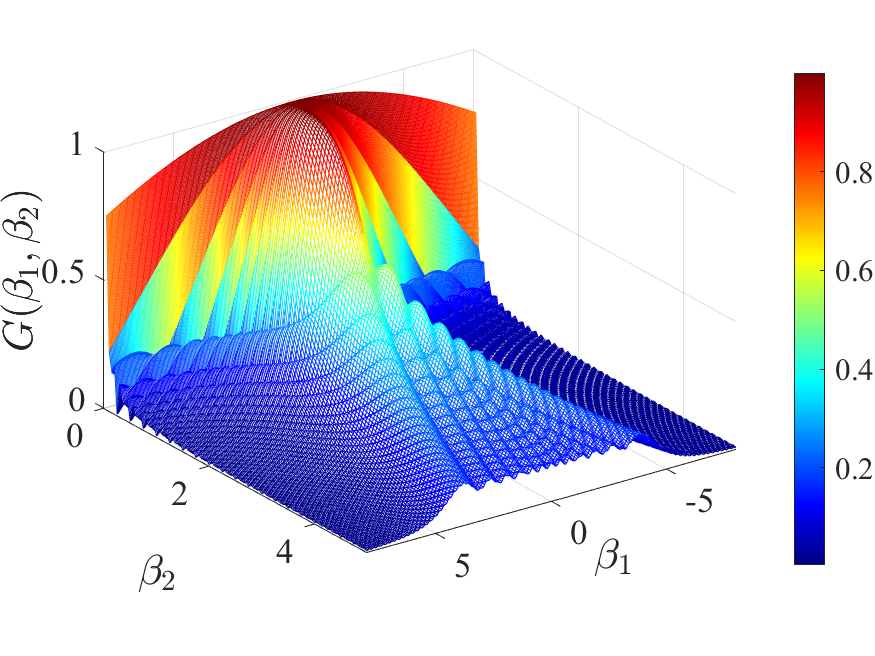}
		\caption{3D Illustration.}
		\label{3D_Illustration}
	\end{subfigure}
	\begin{subfigure}[b]{0.49\linewidth}
		\includegraphics[width=1\linewidth]{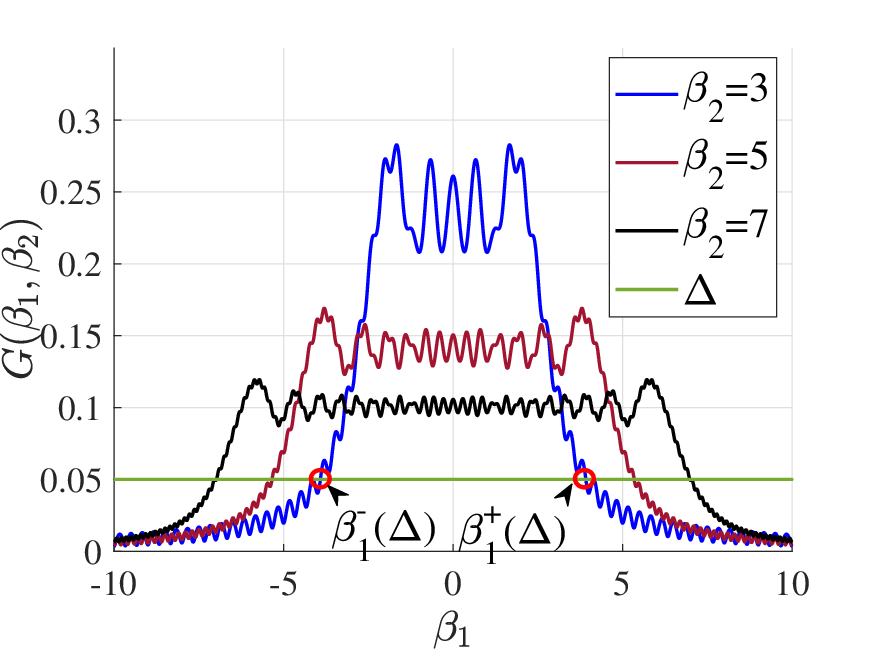}
		\caption{1D Illustration.}
		\label{1D_Illustration}
	\end{subfigure}
	\caption{Illustrations of function $G(\beta_1,\beta_2)$ v.s. $\beta_1$ and $\beta_2$.}
	\label{G_vs_beta1}
	\vspace{-5pt}
\end{figure}

As shown in \eqref{SecuCond}, the secure-transmission condition of Bob is jointly affected by the path-loss ratio and the mixed-field \emph{channel-beam correlation}, both of which are affected by the locations of Bob and Eve. Specifically, a larger angular difference between Bob and Eve (i.e., lower channel-beam correlation) and a larger path-loss ratio generally lead to less information leakage and hence improved secrecy rates. To obtain further insights, we make the following definition.

\begin{definition}\rm 
	The correlation between a near-field steering vectors  $\mathbf{b}^H\left(\theta_p,r_p\right)$ and a far-field steering vectors $\mathbf{a}\left(\theta_q\right)$ is defined as $\eta(\theta_q;\theta_p,r_p) = \left|\mathbf{b}^{H}\left(\theta_p,r_p\right)\mathbf{a}\left(\theta_q\right)\right|$.
\end{definition}

\begin{lemma} \label{correlation} \rm
    The mixed-field correlation function $\eta(\theta_q;\theta_p,r_p)$ can be approximated as \cite{Zhou_Cong2024}
    \begin{align}
        \eta(\theta_q;\theta_p,r_p) &\approx \left|\frac{\hat{C}(\beta_1,\beta_2)+j\hat{S}(\beta_1,\beta_2)}{2\beta_2}\right|\\
        & \triangleq G(\beta_1,\beta_2),
    \end{align}
    	where $\beta_1 = (\theta_p-\theta_q)/\sqrt{\frac{d(1-\theta^2_p)}{r_p}}$ and $ \beta_2 = \frac{N}{2}\sqrt{\frac{d(1-\theta^2_p)}{r_p}}$. Besides, $\hat{C}(\beta_1,\beta_2)=C(\beta_1+\beta_2)-C(\beta_1-\beta_2)$ and  $\hat{S}(\beta_1,\beta_2)=S(\beta_1+\beta_2)-S(\beta_1-\beta_2)$, in which $C(\beta)=\int_{0}^{\beta}\cos(\frac{\pi}{2}t^2){\rm d}t$ and $S(\beta)=\int_{0}^{\beta}\sin(\frac{\pi}{2}t^2){\rm d}t$ are Fresnel integrals.
\end{lemma}
\begin{proof}
    The proof is similar to that in [Theorem 1 \cite{YZhang2023}] and hence is omitted for brevity.
\end{proof}

\begin{remark} \label{propertyofcorre}  \rm  
    (Useful properties of function $G(\cdot)$) To draw useful insights, we first illustrate in Fig. \ref{G_vs_beta1} the function $G(\beta_1,\beta_2)$ versus (v.s.) $\beta_1$ and $\beta_2$. Several important properties of function $G(\cdot)$ are summarized as follows\cite{YZhang2023}.
\begin{itemize}
    \item As shown in Fig. \ref{G_vs_beta1}(a), for fixed $\beta_2$, $G(\beta_1,\beta_2)$ has a significantly large value with slight fluctuations within a specific interval, which is symmetric w.r.t. $\beta_1=0$. Furthermore, this interval gradually expands as $\beta_2$ increases. When outside this interval, $G(\beta_1,\beta_2)$ gradually decreases as $|\beta_1|$ increases.
    \item As shown in Fig. \ref{G_vs_beta1}(b), given $\beta_2$, there exist two solutions for $\beta_1$ to the equation $ G(\beta_1,\beta_2)=\Delta$, when $\Delta<G(0,\beta_2)$. These two solutions are denoted as $\beta^{\rm -}_1(\!\Delta\!)<0$ and $\beta^{\rm +}_1(\!\Delta\!)>0$, respectively.
	Moreover, we can show that both $|\beta^{\rm -}_1(\!\Delta\!)|$ and $|\beta^{\rm +}_1(\!\Delta\!)|$ are decreasing with  $\Delta$.
\end{itemize}
\end{remark}

\begin{figure*}[t]
	\centering
	\begin{subfigure}[b]{0.3\linewidth}
		\includegraphics[width=1\linewidth]{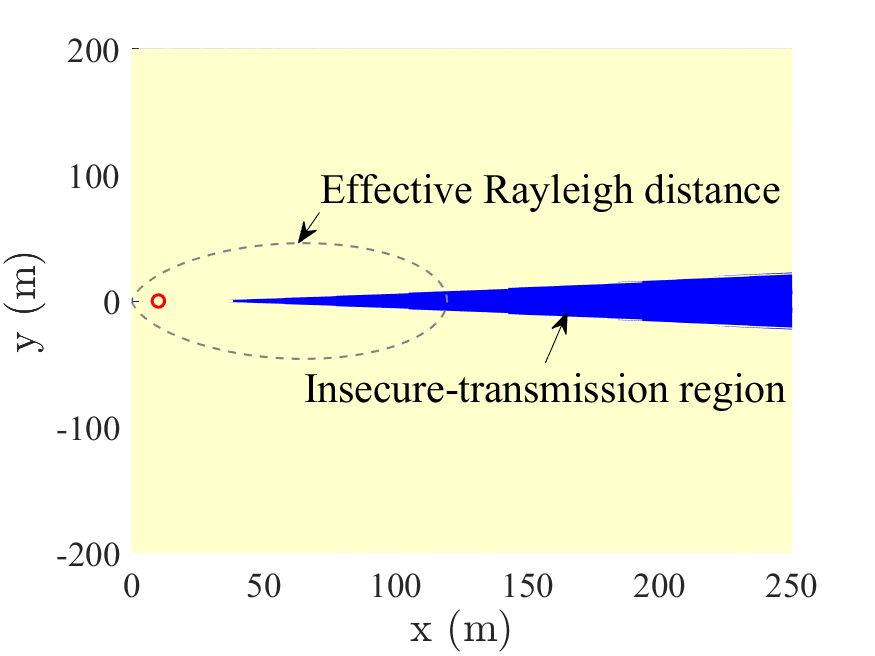}
		\caption{Channels are modeled based on LoS USW.}
		\label{VRvsd1}
	\end{subfigure}
	\hspace{0.2cm}
	\begin{subfigure}[b]{0.3\linewidth}
		\includegraphics[width=1\linewidth]{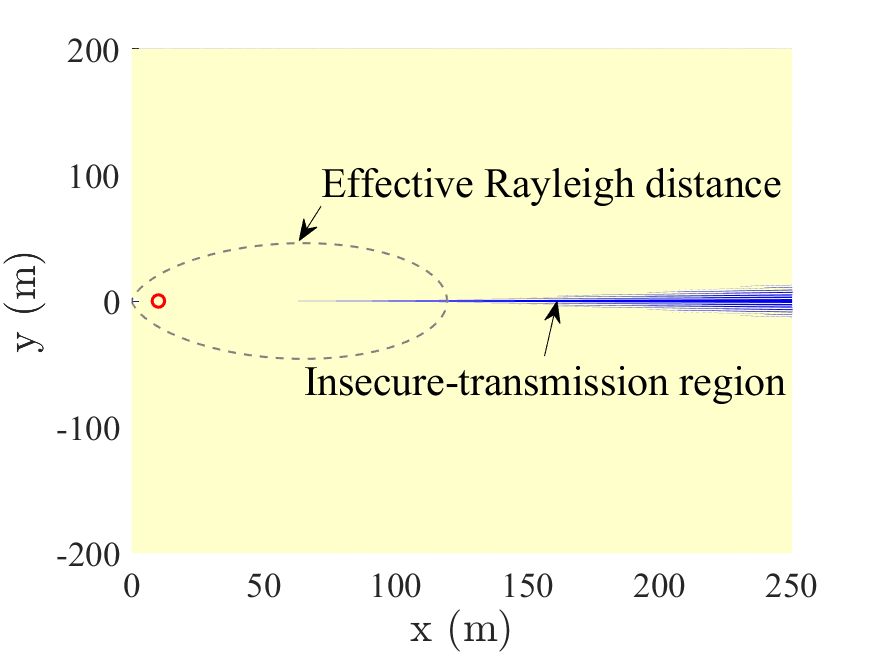}
		\caption{Channels are modeled based on LoS UPW.}
		\label{VRvsSmallN1}
	\end{subfigure}
	\hspace{0.2cm}
	\begin{subfigure}[b]{0.3\linewidth}
		\includegraphics[width=1\linewidth]{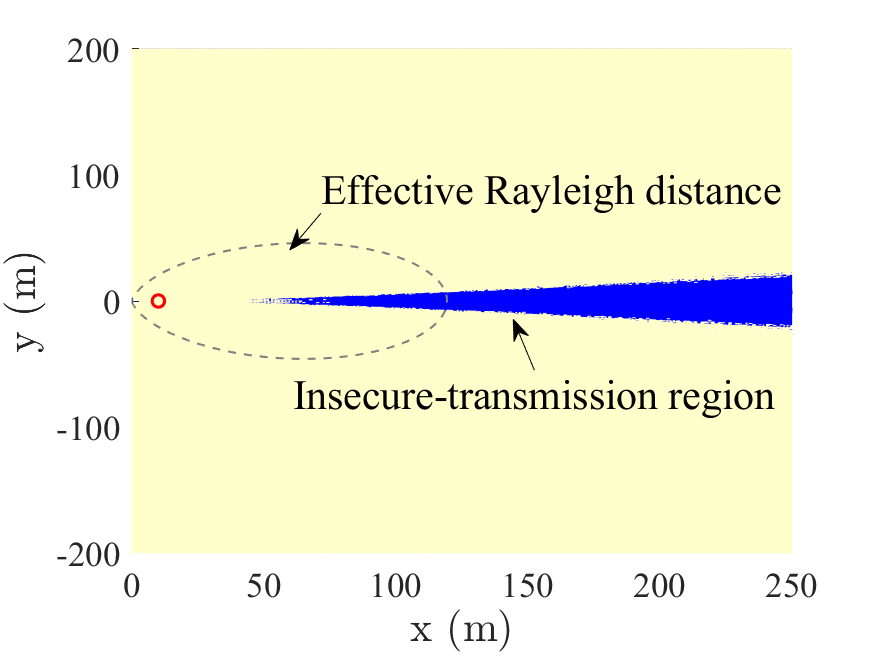}
		\caption{Channels are modeled based on multipath USWs.}
		\label{VRvsmultipath}
	\end{subfigure}
	\caption{Insecure-transmission region of the Bob in the mixed-field PLS system.}
	\label{simulation_result1}
		\vspace{-10pt}
\end{figure*}
Next, we characterize the insecure-transmission region of Bob when the location of Eve is fixed. Moreover, we investigate the effect of the location of Eve on the insecure-transmission region of Bob.

\begin{proposition}\label{SecuAnguRegi} 
\rm Given the location of Eve $(\theta_{\rm E},r_{\rm E})$ (and hence fixed $\beta_2$), the insecure-transmission region of Bob is characterized by
\begin{align}
    \mathcal{A}(\theta_{\rm E},r_{\rm E})=\{(\theta_{\rm B}, r_{\rm B}) \mid \theta_{\rm B} \in \Xi_{\theta}, {\rm and}\; r_{\rm B}>Z \}, \label{tot_insecure_region}
\end{align}
where the angular region (denoted as $\Xi_{\theta}$) is given by
	\begin{align}
		&\Xi_{\theta}\!=\!\left[\!\theta_{\rm E}\!-\!|\beta^{\rm +}_1(\!\Lambda\!)|\sqrt{\frac{d(1\!-\!\theta^2_{\rm E})}{r_{\rm E}}}\!, \theta_{\rm E}\!+\!|\beta^{\rm -}_1(\!\Lambda\!)|\sqrt{\frac{d(1\!-\!\theta^2_{\rm E})}{r_{\rm E}}}\!\right]\!, \label{insecure_thetaB}
	\end{align} 
	with $Z$ being the effective Rayleigh distance and $\Lambda\triangleq\frac{r_{\rm E}}{r_{\rm B}}$. In addition, $\beta^{\rm -}_1(\!\Lambda\!)<0$ and $\beta^{\rm +}_1(\!\Lambda\!)>0$ are the two solutions for $\beta_1$ such that $G(\beta_1,\beta_2)=\Lambda$, respectively. 
\end{proposition}
\begin{proof}
Please refer to Appendix \ref{insecure_thetaB_appendix}. 
\end{proof}

\begin{remark} \label{insecure_region_analyse}\rm 
Several interesting insights into the insecure-transmission angular region are drawn as follows.
\begin{itemize}
\item (Mixed-field v.s. far-field) For the one-Bob-one-Eve system, the insecure-transmission angular region of Bob in the mixed-field PLS case is generally \emph{larger} than that in the far-field PLS case. Specifically, for the mixed-field case, even when the far-field Bob is not located at the same angle with Eve, it can still suffer strong power intercepting (and hence insecure data transmission) if their angle difference is smaller than a threshold due to the energy-spread effect. This is in sharp contrast to the purely far-field case, where the data transmissions to Bob are insecure only when Bob locates at the same angle with Eve \cite{Anaya_Lopez2022}. 
\item (Effect of the location of Eve) \textbf{Proposition \ref{SecuAnguRegi}} indicates that the insecure-transmission angular region of Bob, i.e., $\Xi_{\theta}$, gets narrower when the spatial angle of Eve becomes farther away from the boresight angle (i.e., $\theta_{\rm E}=0$) and/or the distance of Eve becomes longer, which is expected due to the weaker energy-spread effect. 
\item (Effect of the distance of Bob) Based on \textbf{Remark \ref{propertyofcorre}}, when the location of Eve $(\theta_{\rm E},r_{\rm E})$ is given, the insecure-transmission angular region of Bob expands when the distance of Bob is longer, due to the more severe energy-spread effect of the far-field beam in the near-field.
\end{itemize}
\end{remark}

\begin{example}\label{example_secureregion}\rm
    (Insecure-transmission region of Bob) In Fig. \ref{simulation_result1}, we present a concrete example to show the insecure-transmission region of Bob (in blue) for the mixed-field PLS system with $N=256$, $f=30$ GHz and Eve located at (0 rad, 10 m). The channels of Bob and Eve are modeled based on the accurate USW (with UPW being a special case of USW) in Fig. \ref{simulation_result1}(a), while they are approximately modeled based on UPW in Fig. \ref{simulation_result1}(b). It is observed that the insecure-transmission region of Bob in Fig. \ref{simulation_result1}(a) is much larger than that in Fig. \ref{simulation_result1}(b) due to the energy-spread effect in mixed-field PLS systems. This thus necessitates accurate USW modelling for mixed-field PLS systems, since the UPW approximation under-estimates the insecure-transmission region. In addition, we also plot the insecure-transmission region of Bob in Fig. \ref{simulation_result1}(c) for the case where the more general multipath channel model is considered with a Rician factor of 10 dB. The insecure-transmission region obtained based on the multipath channel model is similar that of the simplified LoS channel model (see Fig. \ref{simulation_result1}(a)), which validates the effectiveness of extending the analytical results from the LoS channel case to the multipath channel case. 
\end{example}

\begin{remark} \label{othercase}\rm 
	(Other system setups)
	Three additional system setups are discussed as follows. 
	\begin{itemize}
		\item Both Eve and Bob are located in the far-field: In this scenario, the data transmission to Bob becomes insecure only when Eve locates at the same angle with Bob, due to the sharp beam generated by the XL-array \cite{Boqun_Zhao2024, Anaya_Lopez2022}. 
		\item Both Eve and Bob are in the near-field: In this scenario, when both Bob and Eve are located within the near-field region with a significant energy-focusing effect, the data transmissions to Bob become insecure only when Eve is close to Bob in the range domain\cite{Zhang_Zheng2024}. However, as the distance of Bob increases, the beam depth gradually expands, leading to a weakened energy-focusing effect. This thus results in the increased information leakage to Eve and gives rise to a small angular region where the transmission becomes insecure, as shown in Fig. \ref{simulation_result1}.
		\item Eve is in the far-field, while Bob is in the near-field: In this scenario, the data transmissions to Bob remain secure across all angular spaces, since Eve suffers more severe path loss and has lower beam gain compared to Bob \cite{Zhang_Zheng2024}. 
	\end{itemize}
\end{remark}

\begin{figure}[t]
\centering
\includegraphics[width=3in]{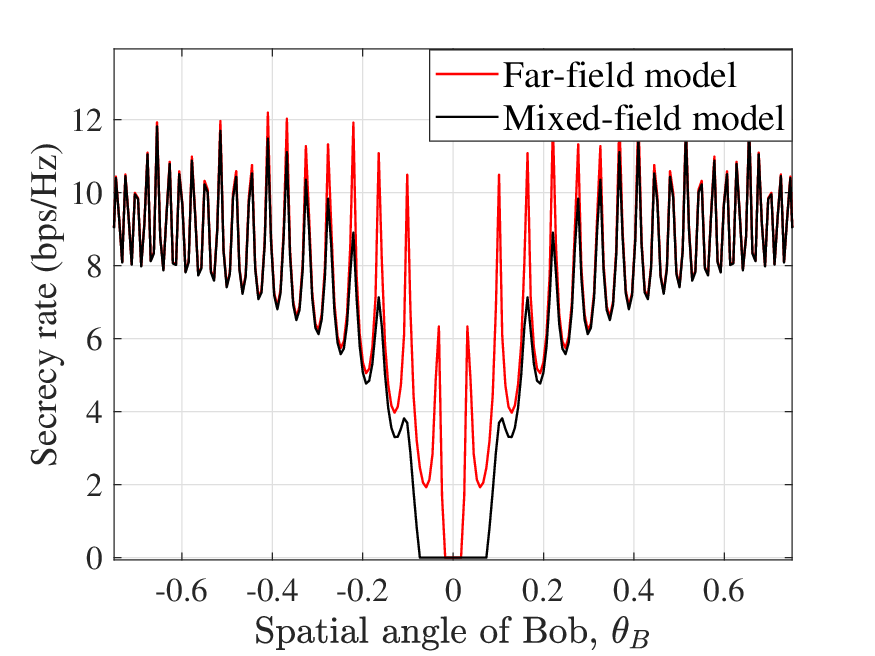}
\caption{Secrecy rate of Bob in the one-Bob-one-Eve system.}
\label{secrecy_rate_oneBob}
	\vspace{-10pt}
\end{figure}

\begin{example}\rm
 (Secrecy rate of Bob) In Fig. \ref{secrecy_rate_oneBob}, we show the secrecy rate of Bob v.s. its spatial angle under the same system setup in Example \ref{example_secureregion}. Two schemes are considered where the channels of Bob and Eve are accurately modeled based on USW (referred to as the mixed-field case) and approximately modeled by UPWs (referred to as the far-field case). It is observed that when the Bob-Eve angular difference decreases, their channel correlation increases, which allows Eve to intercept more information transmitted to Bob, resulting in reduced secrecy rate. Next, the secrecy rate of Bob in the mixed-field system is generally smaller than the one when the channel of Eve is approximately modeled by far-field UPWs. This indicates that it is necessary to judiciously design communication techniques for the mixed-field system. 
\end{example}

\vspace{-10pt}
\subsection{Case II: Two-Bob-One-Eve System}
Next, we consider a two-Bob-one-Eve system\footnote{Since each Eve intercepts information independently in the multi-Eve case, only the Eve with the highest eavesdropping rate needs to be considered. Therefore, the insights obtained from the one-Eve system are applicable to the multi-Eve system. Moreover, the analysis of the two-Bob system can be extended to a multi-Bob scenario, in which the effect of the interference from multiple Bobs to the Eve, as well as the mutual interference among Bobs, are taken into account.}. Note that, compared to the one-Bob-one-Eve PLS system, Bob 2 not only acts as an additional eavesdropped target for Eve, but also serves as an interference source for both Bob 1 and Eve. In particular, its interference to Eve can be used to effectively mitigate the eavesdropping rate of Eve, thereby introducing new effects on the secure-transmission region and the sum-secrecy-rate performance.

For the two-Bob-one-Eve system, when the MRT-based beamforing is applied, the achievable-secrecy-rates of Bob 1 and Bob 2 are respectively given by 
\begin{align}
    R_{{\rm S},1}&\!=\!\left[\!\log _2\!\left(1\!+\!\frac{N P_{{\rm B},1}|h_{{\rm B},1}|^2}{NP_{{\rm B},2}|h_{{\rm B},1}|^2\left|\mathbf{a}^H(\theta_{{\rm B},1})\mathbf{a}(\theta_{{\rm B},2})\right|^2+\sigma^2}\right)\! \right. \nonumber\\
     &\left.-\!\log _2\!\left(\!1\!+\!\frac{N P_{{\rm B},1}|h_{{\rm E}}|^2\left|\mathbf{b}^H(\theta_{{\rm E}},r_{{\rm E}})\mathbf{a}(\theta_{{\rm B},1})\right|^2}{NP_{{\rm B},2}|h_{{\rm E}}|^2\left|\mathbf{b}^H(\theta_{{\rm E}},r_{{\rm E}})\mathbf{a}(\theta_{{\rm B},2})\right|^2+\sigma^2}\!\right)\!\right]^{+}\!,\label{TwoBobRs1}\\
     R_{{\rm S},2}&\!=\!\left[ \!\log _2\!\left(\!1\!+\!\frac{N P_{{\rm B},2}|h_{{\rm B},2}|^2}{NP_{{\rm B},1}|h_{{\rm B},2}|^2\left|\mathbf{a}^H(\theta_{{\rm B},2})\mathbf{a}(\theta_{{\rm B},1})\right|^2+\sigma^2}\!\right)\! \right. \nonumber\\
     &\left.-\!\log _2\!\left(\!1\!+\!\frac{N P_{{\rm B},2}|h_{{\rm E}}|^2\left|\mathbf{b}^H(\theta_{{\rm E}},r_{{\rm E}})\mathbf{a}(\theta_{{\rm B},2})\right|^2}{NP_{{\rm B},1}|h_{{\rm E}}|^2\left|\mathbf{b}^H(\theta_{{\rm E}},r_{{\rm E}})\mathbf{a}(\theta_{{\rm B},1})\right|^2+\sigma^2}\!\right)\!\right]^{+}\!. \label{TwoBobRs2}
\end{align}
Then, Problem (P2) reduces to
\begin{align}
    \text{(P3)}\;\; \max_{\mathbf{p}} \quad & R_{{\rm S},1}+R_{{\rm S},2} \label{P2sumrate}\\
    \text{s.t.}\;\;\;\;&P_{{\rm B},1}+P_{{\rm B},2} \leq P_{\rm tot}. \label{P2powercons}
\end{align}
The optimal solution to Problem (P3) can be obtained by a one-dimensional search for $P_{{\rm B},1}$ and $P_{{\rm B},2}$. To obtain useful insights, we first characterize the secure-transmission region of the considered system and then obtain approximate solution to Problem (P3) in closed form under high-SNR condtions. Note that, compared with the performance analysis for the one-Bob-One-Eve system, the secure-transmission region in this case is more complicated since the secrecy rates of both two Bobs are jointly affected by the power allocation of the BS and the locations of Eve, Bob 1 and Bob 2 (see \eqref{TwoBobRs1} and \eqref{TwoBobRs2}).

\subsubsection{Insecure-transmission region of Bobs}
In the following, we characterize the insecure-transmission region of Bob 1, while the same analytical method can be applied for that of Bob 2. To this end, we first obtain the secure-transmission condition of Bob 1, which is given as follows.

\begin{lemma}  \label{securecond_multiuser1}\rm(Secure-transmission condition of Bob 1)
   In the two-Bob-one-Eve system, the data transmissions to Bob 1 are secure when
    \begin{equation}\label{seccon_for_multiuser}
    \begin{cases} 
        &\left|\mathbf{b}^H\left(\theta_{\rm E},r_{\rm E}\right)\mathbf{a}\left(\theta_{{\rm B},1}\right)\right|^2\!< \! \frac{r^2_{\rm E}\!+\!\frac{\beta NP_{{\rm B},2}}{\sigma^2}\!\left|\mathbf{b}^H\(\theta_{\rm E},r_{\rm E}\)\mathbf{a}\(\theta_{{\rm B},2}\)\right|^2}{r^2_{{\rm B},1}}, \\ &\qquad\qquad\qquad\qquad\qquad\qquad\qquad\qquad {\rm if}\;  \theta_{{\rm B},1}\neq\theta_{{\rm B},2},
\\
      &\left|\mathbf{b}^H\left(\theta_{\rm E},r_{\rm E}\right)\mathbf{a}\left(\theta_{{\rm B},1}\right)\right|^2<\frac{r^2_{\rm E}}{r^2_{{\rm B},1}}, \;\;\;\;\qquad {\rm if} \; \theta_{{\rm B},1}=\theta_{{\rm B},2}. 
   \end{cases}
   \end{equation}
\end{lemma}
\begin{proof}
   Please refer to Appendix \ref{proof_resecure_condition}.
\end{proof}

\begin{remark} \rm Several key insights into the secure-transmission condition of Bob 1 are drawn as follows when Bob 2 is added into the system. 
\begin{itemize}
    \item (Effect of Bob 2) \textbf{Lemma \ref{securecond_multiuser1}} indicates that, when $\theta_{{\rm B},1}\neq\theta_{{\rm B},2}$, an additional interference term from Bob 2 to Eve, i.e., $\frac{\beta NP_{{\rm B},2}}{\sigma^2}\!\left|\mathbf{b}^H\left(\theta_{\rm E},r_{\rm E}\right)\mathbf{a}\left(\theta_{{\rm B},2}\right)\right|^2$, is introduced in the secure-transmission condition of Bob 1 in the two-Bob system, as compared to that in the one-Bob system given in \eqref{SecuCond}. Moreover, it makes the data transmissions to Bob 1 secure when the interference from Bob 2 to Eve is sufficiently large. On the other hand, when $\theta_{{\rm B},1}=\theta_{{\rm B},2}$, the secure-transmission condition of Bob 1 in \eqref{seccon_for_multiuser} reduces to that in \eqref{SecuCond} for the one-Bob system. In this case, Bob 2 does not affect too much the secure-transmission condition of Bob 1.
    \item (Effect of power allocation) The interferences from Bob 2 to Bob 1 and Eve in \eqref{TwoBobRs1} are affected by the power allocated to Bob 2. Generally, a larger power allocated to Bob 2 results in both enhanced interferences to Bob 1 and Eve, which not only reduces the eavesdropping rate of Eve for intercepting the data transmitted to Bob 1, but also reduces the information rate of Bob 1. Therefore, there generally exists a trade-off in power allocation between improving the information rate of Bob 1 and suppressing the eavesdropping rate on Bob 1.
\end{itemize}
\end{remark}

Next, we further characterize the insecure-transmission region of Bob 1. According to \textbf{Lemma \ref{securecond_multiuser1}}, when $\theta_{{\rm B},1} =\theta_{{\rm B},2}$,  the insecure-transmission region of Bob 1 coincides with that in the one-Bob-one-Eve scenario. Therefore, we only present the result for the case where $\theta_{{\rm B},1} \neq \theta_{{\rm B},2}$.

\begin{proposition} \label{insecure_region_2Bob}
\rm Given the locations of Eve $(\theta_{\rm E},r_{\rm E})$ (and hence fixed $\beta_2$) and Bob 2 $(\theta_{{\rm B},2},r_{{\rm B},2})$, when $\theta_{{\rm B},1} \neq \theta_{{\rm B},2}$, the insecure-transmission region of Bob 1 is characterized by
\begin{align}
    \mathcal{A}^{\prime}(\theta_{\rm E},r_{\rm E})=\{(\theta_{{\rm B},1}, r_{{\rm B},1}) \mid \theta_{{\rm B},1}\! \in \! \Xi^{\prime}_{\theta}, {\rm and}\; r_{{\rm B},1} \!> \!Z\}, \label{tot_insecure_region_B1}
\end{align}
where the angular region (denoted as $\Xi^{\prime}_{\theta}$) is given by
	\begin{align}
		&\Xi^{\prime}_{\theta}\!=\!\left[\!\theta_{\rm E}\!-\!|\beta^{\rm +}_1(\!\Lambda^{\prime}\!)|\sqrt{\frac{d(1\!-\!\theta^2_{\rm E})}{r_{\rm E}}}\!, \theta_{\rm E}\!+\!|\beta^{\rm -}_1(\!\Lambda^{\prime}\!)|\sqrt{\frac{d(1\!-\!\theta^2_{\rm E})}{r_{\rm E}}}\!\right]\!, \label{insecure_thetaB1}
	\end{align} 
	with $\Lambda^{\prime}\triangleq\sqrt{\left(\frac{\beta NP_{{\rm B},2}}{\sigma^2}\!\left|\mathbf{b}^H\(\theta_{\rm E},r_{\rm E}\)\mathbf{a}\(\theta_{{\rm B},2}\)\right|^2\!+\!r^2_{\rm E}\right)/r^2_{{\rm B},1}}$. In addition, $\beta^{\rm -}_1(\!\Lambda^{\prime}\!)<0$ and $\beta^{\rm +}_1(\!\Lambda^{\prime}\!)>0$ are the two solutions for $\beta_1$ such that $G(\beta_1,\beta_2)=\Lambda^{\prime}$, respectively. 
\end{proposition}
\begin{proof}
   The proof is similar to that in \textbf{Proposition \ref{SecuAnguRegi}} and hence is omitted for brevity.
\end{proof}

Note that, compared with the insecure-transmission region of Bob in the one-Bob system, Bob 2 introduces additional interferences to Eve, which is given by $\frac{\beta NP_{{\rm B},2}}{\sigma^2}\left|\mathbf{b}^H\left(\theta_{\rm E},r_{\rm E}\right)\mathbf{a}\left(\theta_{{\rm B},2}\right)\right|^2$. This results in a larger value of $\Lambda^{\prime}$ in \eqref{insecure_thetaB1} compared to $\Lambda$ in \eqref{insecure_thetaB} under the one-Bob case, thereby leading to diminished insecure-transmission angular region according to \textbf{Lemma \ref{propertyofcorre}}. Moreover, this interference term is influenced by the location of Bob 2, i.e., $(\theta_{{\rm B},2}, r_{{\rm B},2})$, and the power allocated to it, i.e., $P_{{\rm B},2}$.

\begin{remark} \label{analyzeforBob1}\rm
With the introduction of Bob 2 and $\theta_{{\rm B},1} \neq \theta_{{\rm B},2}$, several interesting insights into the insecure-transmission region of Bob 1 are drawn as follows.
    \begin{itemize}
    \item (Effect of the location of Bob 2) The insecure-transmission angular region of Bob 1 gets narrower as the angular difference between Bob 2 and Eve decreases. This is due to the increased interference power from Bob 2 to Eve, which is leveraged to prevent the eavesdropping of Eve on Bob 1.
    \item (Effect of power allocation) When the power allocation to Bob 2 increases, the Bob 2 $\rightarrow$ Eve interference power becomes stronger, thus leading to a smaller insecure-transmission region. 
    \end{itemize}
\end{remark}

\subsubsection{Power allocation and secrecy rate} 
Given the secure-transmission condition for the two-Bob-one-Eve system, we further optimize the power allocation to the two Bobs and obtain the maximum sum-secrecy-rate. It is challenging to obtain the optimal power allocation solution to Problem (P3) in closed form, which, however, can generally be obtained via exhaustive search. To gain useful insights, we further consider two cases. Case I: The data transmission to Bob 1 is insecure, while that to Bob 2 is secure, and Case II: The data transmissions to both Bobs are secure; while the scenario where the data transmissions to both Bobs are insecure is not considered, as allocating power to them in such a scenario would be ineffective. Under the high-SNR assumption, we present the corresponding approximate closed-form solutions for the power allocation to the two Bobs in the above considered cases as follows.

\begin{proposition}\label{optimal_poweralloc} \rm
In high-SNR regimes, when $\theta_{\rm B,1} \neq \theta_{\rm B,2}$, the solution for Problem (P3) is given as follows.
   
\textbf{Case I:} In this case, the solution to the power allocation is given by
    \begin{equation}\label{powerallc_case2}
    \begin{cases} 
      P^*_{\rm B,1} &= 0, \\
      P^*_{\rm B,2} &= P_{\rm tot}.
   \end{cases}
   \end{equation}

\textbf{Case II:} In this case, the optimized power allocation can be further categorized into the following two cases.
\begin{itemize}
    \item When $\left|\!\mathbf{b}^H\!\left(\!\theta_{\rm E}, r_{\rm E}\!\right)\!\mathbf{a}\!\left(\!\theta_{\rm B,2}\!\right)\!\right|\!\rightarrow\!0$ and $\left|\!\mathbf{b}^H\!\left(\!\theta_{\rm E}, r_{\rm E}\!\right)\!\mathbf{a}\!\left(\!\theta_{\rm B,1}\!\right)\!\right|\!\rightarrow\!0$, no data is intercepted by Eve. As such, the optimal solution for the power allocation is given by
        \begin{equation}
    	\begin{cases} \label{powerallca_case1}
    		P^*_{{\rm B},1} &= \frac{1}{\gamma}-\frac{\sigma^2}{N|h_{{\rm B},1}|^2}, \\
    		P^*_{{\rm B},2} &= \frac{1}{\gamma}-\frac{\sigma^2}{N|h_{{\rm B},2}|^2},
    	\end{cases}
    \end{equation}
    where $\gamma=2/\left(P+\frac{\sigma^2}{N|h_{{\rm B},1}|^2}+\frac{\sigma^2}{N|h_{{\rm B},2}|^2}\right)$.
    \item When $\left|\!\mathbf{b}^H\!\left(\!\theta_{\rm E}, r_{\rm E}\!\right)\!\mathbf{a}\!\left(\!\theta_{\rm B,2}\!\right)\!\right|$ and $\left|\!\mathbf{b}^H\!\left(\!\theta_{\rm E}, r_{\rm E}\!\right)\!\mathbf{a}\!\left(\!\theta_{\rm B,1}\!\right)\!\right|$ are sufficiently large such that the noise power at Eve can be neglected, the approximate solution for the optimal power allocation is given by 
    \begin{equation}\label{poweralloc_case3}
    \begin{cases} 
      P^*_{\rm B,1} \!\!\!\!\!&\approx \hat{P}_{\rm B,1} \triangleq \frac{P_{\rm tot}\left|\mathbf{b}^H\left(\theta_{\rm E},r_{\rm E}\right)\mathbf{a}\left(\theta_{\rm B,2}\right)\right|}{\left|\mathbf{b}^H\left(\theta_{\rm E},r_{\rm E}\right)\mathbf{a}\left(\theta_{\rm B,1}\right)\right|+\left|\mathbf{b}^H\left(\theta_{\rm E},r_{\rm E}\right)\mathbf{a}\left(\theta_{\rm B,2}\right)\right|},\\
      P^*_{\rm B,2} \!\!\!\!\!&\approx \hat{P}_{\rm B,2} \triangleq \frac{P_{\rm tot}\left|\mathbf{b}^H\left(\theta_{\rm E},r_{\rm E}\right)\mathbf{a}\left(\theta_{\rm B,1}\right)\right|}{\left|\mathbf{b}^H\left(\theta_{\rm E},r_{\rm E}\right)\mathbf{a}\left(\theta_{\rm B,1}\right)\right|+\left|\mathbf{b}^H\left(\theta_{\rm E},r_{\rm E}\right)\mathbf{a}\left(\theta_{\rm B,2}\right)\right|}.
   \end{cases}
   \end{equation}
\end{itemize}
\end{proposition}
\begin{proof}
 Please refer to Appendix \ref{solutioncaseIII}
\end{proof}
For Case II,  when $\left|\!\mathbf{b}^H\!\left(\!\theta_{\rm E}, r_{\rm E}\!\right)\!\mathbf{a}\!\left(\!\theta_{\rm B,2}\!\right)\!\right|$ and $\left|\!\mathbf{b}^H\!\left(\!\theta_{\rm E}, r_{\rm E}\!\right)\!\mathbf{a}\!\left(\!\theta_{\rm B,1}\!\right)\!\right|$ are sufficiently large, the power allocation for the two Bobs is determined by the channel correlation between the Bobs and Eve. Specifically, taking Bob 2 as an example, as the channel correlation between Bob 2 and Eve decreases, the information leakage from Bob 2 to Eve diminishes; thus, a larger power should be allocated to Bob 2 to enhance sum-secrecy-rate.

Note that it is difficult to obtain the optimal power allocation solution for the case of $\theta_{\rm B,1} = \theta_{\rm B,2}$ in closed form. Therefore, we present numerical results for it in the following example.

\begin{figure*}[t]
    \centering
     \begin{subfigure}[b]{0.3\linewidth}
     \centering
    \includegraphics[width=1\linewidth,height=4.5cm]{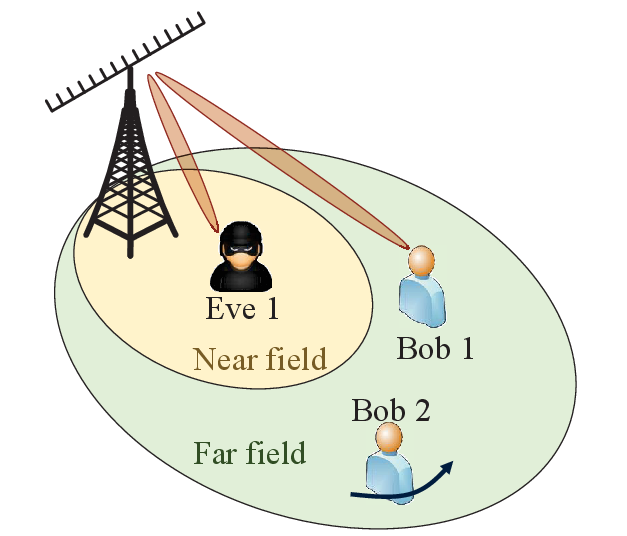}
        \caption{Two-Bob-one-Eve system.}
        \label{system_two user}
    \end{subfigure}
    \hspace{-0.5cm}
    \begin{subfigure}[b]{0.3\linewidth}
    \centering
        \includegraphics[width=1.1\linewidth]{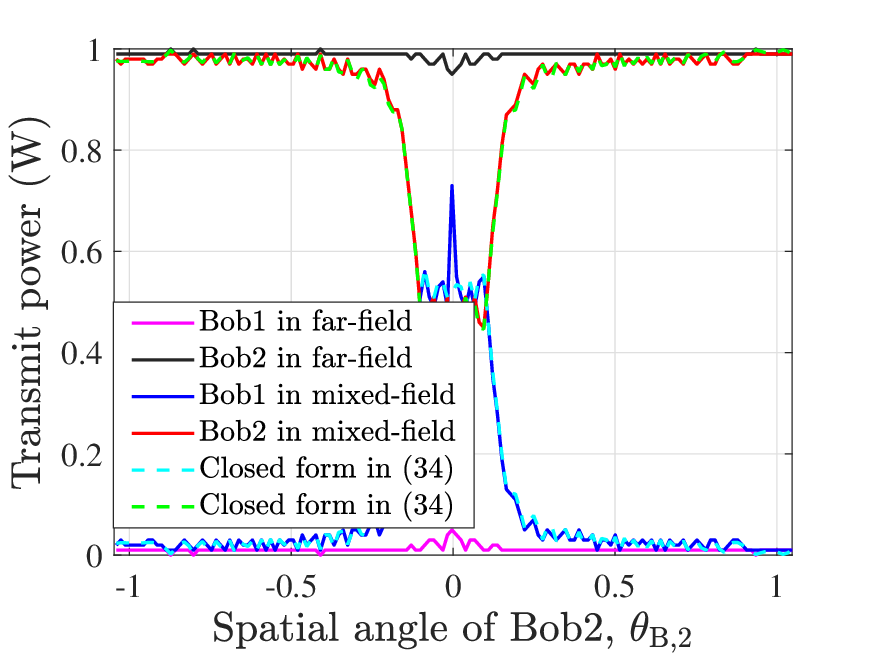}
        \caption{Power allocation v.s. angle $\theta_{{\rm B},2}$.}
        \label{powerallocation}
    \end{subfigure}
    \hspace{0.1cm}
    \begin{subfigure}[b]{0.3\linewidth}
    \centering
        \includegraphics[width=1.1\linewidth]{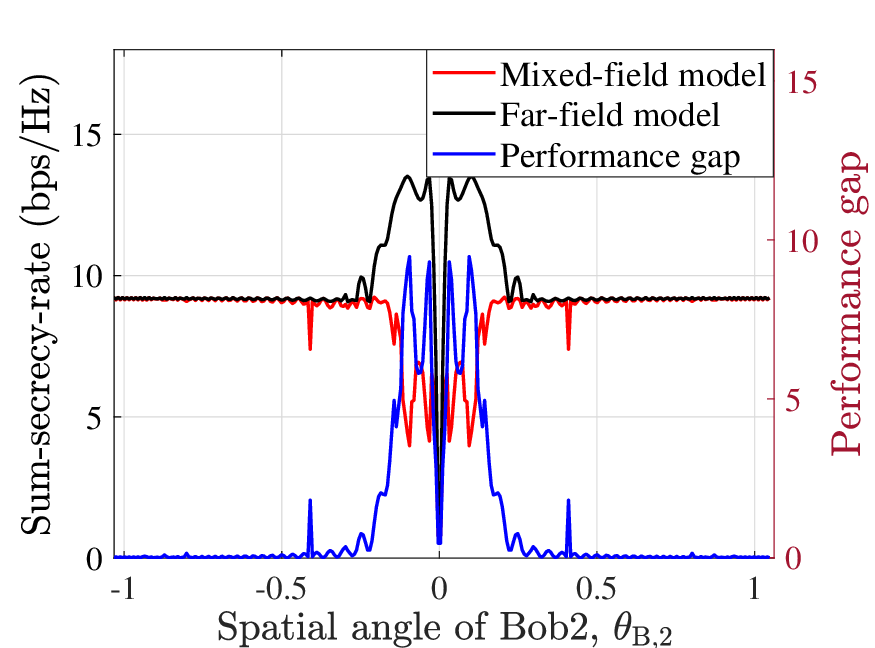}
        \caption{Sum-secrecy-rate v.s. angle $\theta_{{\rm B},2}$.}
        \label{secrecyrate_twouser_vsangle2}
    \end{subfigure}
    \caption{Illustration of the power allocation designs based on the mixed-field and far-field channel models, and their corresponding performance.}
    \label{simulation_result_poweralloc}
    	\vspace{-7pt}
\end{figure*}
\begin{example}\label{secrecyrate_twoBob}\rm 
(Power allocation and secrecy rate) A mixed-field PLS system is considered, where $N=256$, $f=30$ GHz, Eve is located at (0 rad, 5 m), Bob 1 is located at (0 rad, 229 m), and the distance of Bob 2 is $r_{\rm B,2}=164$ m. We plot the optimal power allocation for Bob 1 and Bob 2 v.s. the spatial angle of Bob 2 in Fig. \ref{simulation_result_poweralloc}(b), and the corresponding achievable sum-secrecy-rate in Fig. \ref{simulation_result_poweralloc}(c). Two schemes are considered where the channels of Bob and Eve are accurately modeled based on USW (referred to as the mixed-field case) and approximately modeled by UPWs (referred to as the far-field case). First, it is observed from Fig. \ref{simulation_result_poweralloc}(b) that the power allocation in the mixed-field case obtained through the exhaustive search method matches well with the optimized power allocation obtained in \eqref{poweralloc_case3} for the case where $\theta_{\rm B,2}\neq\theta_{\rm B,1}$. Next, it is observed from Fig. \ref{simulation_result_poweralloc}(b) that in the mixed-field case, as the angular difference between Bob 2 and Eve decreases, a higher power should be allocated to Bob 1 due to the increasing interference power from Bob 2 to Eve, which effectively prevents the eavesdropping of Eve on Bob 1. In contrast, in far-field case, nearly no power is allocated to Bob 1 unless Bob 2 and Eve are at the same spatial angle. Therefore, as shown in Fig. \ref{simulation_result_poweralloc}(c), when Bob 2 is located within a certain angular region near Eve (i.e., around 0 rad), the far-field case exhibits a noticeable secrecy performance gap compared to the mixed-field case. This is due to inefficient power allocation in the far-field case, where Bob 2 can not effectively mitigate eavesdropping capacity of Eve unless they are located at the same angle. Last, when Bob 2 and Eve are at the same spatial angle, i.e., $\theta_{\rm B,2}=\theta_{\rm B,1}$, the significant interference power between Bob 1 and Bob 2 results in poor secrecy rate performance, as shown in Fig. \ref{simulation_result_poweralloc}(c). 
\end{example}

\section{Proposed Algorithm for Solving (P2)}
In this section, we propose an efficient algorithm to solve Problem (P2) for the general case consisting of multiple Eves and multiple Bobs, while the extension to solve (P1) will be discussed in Remark \ref{AlgorithmForhybrid}. 

\subsection{Proposed Algorithm}
First, Problem (P2) can be equivalently reformulated as
\begin{align}
    \text{(P4)}\;\; \max_{\mathbf{p}} \quad &\sum^K_{k=1} \left(R_{\rm B,k}-\max \left\{ R_{\rm E,m,k}, m \in \mathcal{M} \right\}\right) \label{P1sumrate2}\\
    \text{s.t.}\;\;\;\;  &\eqref{P1powercons}. \nonumber
\end{align}
\begin{lemma} \rm
    The solution to Problem (P4) also solves Problem (P2).
\end{lemma}
\begin{proof}
    Problems (P2) and (P4) have the same optimal solution due to the fact that the value of each summation term in the objective function of both problems, i.e., $R_{\rm B,k}-R_{\rm E,k}, \forall k$, must be non-negative in the optimal solution. This is because if $R_{\rm B,k}-R_{\rm E,k} \leq 0$ for any $k$, we can always increase its value in the objective to zero by setting $P_{\rm B,k}=0$ without violating the power constraints \eqref{P1powercons}, thus completing the proof.
\end{proof}

Problem (P4) is a non-convex optimization problem and is challenging to solve. To address this issue, we first introduce slack variables $\{\tilde{R}_{{\rm B},k}\}$ and $\{\tilde{R}_{{\rm E},k}\}$ such that
\begin{align}
    &\tilde{R}_{{\rm B},k} \leq R_{{\rm B},k}, k \in \mathcal{K},\label{slackRU}\\
    &\tilde{R}_{{\rm E},k} \geq R_{{\rm E},k} = \max \left\{ R_{{\rm E},m,k}, m \in \mathcal{M} \right\}, k \in \mathcal{K},\label{slackRE}
\end{align}
where constraints \eqref{slackRE} can be re-expressed as
\begin{align}
\tilde{R}_{{\rm E},k} \geq R_{{\rm E},m,k}, \forall m,k.\label{slackRE1} 
\end{align}
Then, Problem (P4) can be equivalently reformulated as
\begin{align}
    \text{(P5)}\;\; \max_{\mathbf{p}} \quad &\sum^K_{k=1} \left(\tilde{R}_{{\rm B},k}-\tilde{R}_{{\rm E},k}\right) \label{P1sumrate3}\\
    \text{s.t.}\;\;\;\;  &\eqref{P1powercons}, \eqref{slackRU}, \eqref{slackRE1}.\nonumber
\end{align}
However, Problem (P5) is still a non-convex optimization problem due to the non-convex constraints \eqref{slackRU} and \eqref{slackRE1}. To tackle this issue, we first re-writte these two constraints as
\begin{align}
    &\tilde{R}_{{\rm B},k} \leq \hat{R}_{{\rm B},k}-\check{R}_{{\rm B},k},   \forall k,\label{Ruk2log}\\
    &\tilde{R}_{{\rm E},k} \geq \hat{R}_{{\rm E},m}-\check{R}_{{\rm E},m,k}, \forall m,k, \label{Rek2log}
\end{align}
where $\hat{R}_{{\rm B},k}$, $\check{R}_{{\rm B},k}$, $\hat{R}_{{\rm E},m,k}$, and $\check{R}_{{\rm E},m,k}$ are given by 
$\hat{R}_{{\rm B},k}=\log_2(\sum^K_{i=1}\!P_{{\rm B},i}\!v^2_{{\rm B},k,i}\!+\!1)$, $\check{R}_{{\rm B},k} = \log_2(\sum^K_{i=1,i \neq k}\!P_{{\rm B},i}\!v^2_{{\rm B},k,i}\!+1)$, $\hat{R}_{{\rm E},m} \!= \!\log_2(\!\sum^K_{i=1}P_{{\rm B},i}\!v^2_{{\rm E},m,i}\!+\!1)$, and $\check{R}_{{\rm E},m,k}\!=\!\log_2(\sum^K_{i=1,i \neq k}\!P_{{\rm B},i}\!v^2_{{\rm E},m,i}\!+\!1)$. Herein, $v_{{\rm B},k,i}\triangleq\left|\mathbf{h}^H_{{\rm B},k} \mathbf{w}_{{\rm B},i}\right|/\sigma$ and $v_{{\rm E},m,i}\triangleq\left|\mathbf{h}^H_{{\rm E},m} \mathbf{w}_{{\rm B},i}\right|/\sigma$.

It is observed that $\check{R}_{{\rm B},k}$ in the constraint \eqref{Ruk2log} is a convex function, while $\hat{R}_{{\rm E},m}$ in the constraint \eqref{Rek2log} is a concave function, making this problem being non-convex. In the following, the efficient SCA technique is leveraged to approximate Problem (P5) to be a convex form. 
\begin{lemma} \label{taylor1} \rm
     For the constraint \eqref{Ruk2log}, $\check{R}_{{\rm B},k}$ is a concave function of $P_{{\rm B},i}$. For any local point $\tilde{P}_{{\rm B},i}$, $\check{R}_{{\rm B},k}$ in \eqref{Ruk2log} can be upper-bounded by
    \begin{align}
        \check{R}_{{\rm B},k} &\leq A_{k}+\sum^K_{i=1,i \neq k}B_{k,i} \times \left(P_{{\rm B},i}-\tilde{P}_{{\rm B},i}\right) \triangleq\check{R}^{\rm ub}_{{\rm B},k}, \label{etalb}
    \end{align}
    where the coefficients $A_{k}$ and $B_{k}$ are respectively given by $A_{k}=\log_2\left(\sum^K_{i=1,i \neq k}\tilde{P}_{{\rm B},i}v^2_{{\rm B},k,i}+1\right)$ and $B_{k}=\frac{v^2_{{\rm B},k,i}\log_2(e)}{\sum^K_{j=1,j \neq k}\tilde{P}_{{\rm B},j}v^2_{{\rm B},k,j}+1}$.
\end{lemma}
\begin{lemma} \label{taylor2} \rm
 For the constraint \eqref{Rek2log}, $\hat{R}_{{\rm E},m}$ is a concave function of $P_{{\rm B},i}$. For any local point $\tilde{P}_{{\rm B},i}$,  $\hat{R}_{{\rm E},m}$ in \eqref{Rek2log} can be upper-bounded by
    \begin{align}
        \hat{R}_{{\rm E},m} &\leq C_{m}+\sum^K_{i=1}D_{m,i} \times (P_{{\rm B},i}-\tilde{P}_{{\rm B},i}) \triangleq \hat{R}^{\rm ub}_{{\rm E},m},  \label{gammalb}
    \end{align}
    where the coefficients $C_{m}$ and $D_{m,i}$ are given by $C_{m}=\log_2\left(\sum^K_{i=1}\tilde{P}_{{\rm B},i}v^2_{{\rm E},m,i}+1\right)$ and $D_{m,i}=\frac{v^2_{{\rm E},m,i}\log_2(e)}{\sum^K_{j=1}\tilde{P}_{{\rm B},j}v^2_{{\rm E},m,j}+1}$.
\end{lemma}

According to \textbf{Lemmas \ref{taylor1}} and \textbf{\ref{taylor2}}, by substituting $\check{R}^{\rm ub}_{{\rm B},k}$ in \eqref{etalb} and $\hat{R}^{\rm ub}_{{\rm E},m}$ in \eqref{gammalb} with their corresponding upper bounds, Problem (P5) can be transformed into the following approximate form
\vspace{-5pt}
\begin{align}
    \text{(P6)}\;\max_{\mathbf{p}} \quad &\sum^K_{k=1}\!\left(\tilde{R}_{\rm B,k}\!-\!\tilde{R}_{\rm E,k}\right) \label{P1sumrate4}\\
    \text{s.t.}\;\;\;\;\; &\tilde{R}_{\rm B,k} \leq \hat{R}_{{\rm B},k}\!-\!\check{R}^{\rm ub}_{{\rm B},k}, k \in \mathcal{K},\label{slackRU4}\\
    &\tilde{R}_{\rm E,k}\!\geq\!\hat{R}^{\rm ub}_{{\rm E},m}\!-\!\check{R}_{{\rm E},m,k}, k \!\in\! \mathcal{K}, m \!\in\! \mathcal{M},\\
    &\eqref{P1powercons}.\nonumber
\end{align}
Problem (P6) is a convex optimization problem, which can be efficiently solved by using the standard CVX solvers. As such, the objective value of (P2) can be obtained by solving the convex Problem (P6) iteratively. 

 \begin{remark} \label{AlgorithmForhybrid}\rm
	The proposed algorithm for solving Problem (P2) can also be easily extended to solve Problem (P1). 
	Specifically, we first define the effective channel for Bob $k$ and Eve $m$ as $(\mathbf{h}^{\rm eff}_{{\rm B},k})^H = \mathbf{h}^H_{{\rm B},k}\mathbf{F}_{\rm A}$ and $(\mathbf{h}^{\rm eff}_{{\rm E},k})^H = \mathbf{h}^H_{{\rm E},m}\mathbf{F}_{\rm A}$, respectively. Then, given the effective channels under the MRT-based analog beamforming matrix, Problem (P1) reduces to a digital beamforming design under a total power constraint.
	To address the non-convexity, we first introduce auxiliary variables to reformulate the problem into a more tractable form, and then apply first-order Taylor expansion to approximate the non-convex constraints by convex ones.
	The resulting problem is iteratively solved using the SCA technque\cite{Sohrabi2016}.
\end{remark}

 \subsection{Algorithm Convergence and Complexity Analysis}
 The convergence and complexity analysis for the proposed algorithm is as follows. First, for the solution to Problem (P6) obtained in each iteration, the objective value is no smaller than that in the previous iteration. Thus, the objective value of Problem (P6) is non-decreasing over iterations. Besides, since the objective value has a finite upper bound, the proposed algorithm is guaranteed to converge. Next, the overall algorithm complexity can be characterized as $\mathcal{O}(IK^{3.5})$, where $K$ is the total number of optimization variables (i.e., the number of Bobs), $I$ denotes the number of iterations to achieve convergence. While the complexity of the proposed algorithm in \textbf{Remark \ref{AlgorithmForhybrid}} for Problem (P1) is $\mathcal{O}(IK^{7})$. This indicates that the computational complexity of the proposed algorithm under the low-complexity hybrid architecture is significantly lower than that of the algorithm presented in \textbf{Remark \ref{AlgorithmForhybrid}}.

\begin{figure*}[t]
	\centering
	\begin{subfigure}[b]{0.49\linewidth}
		\centering
		\includegraphics[width=3in]{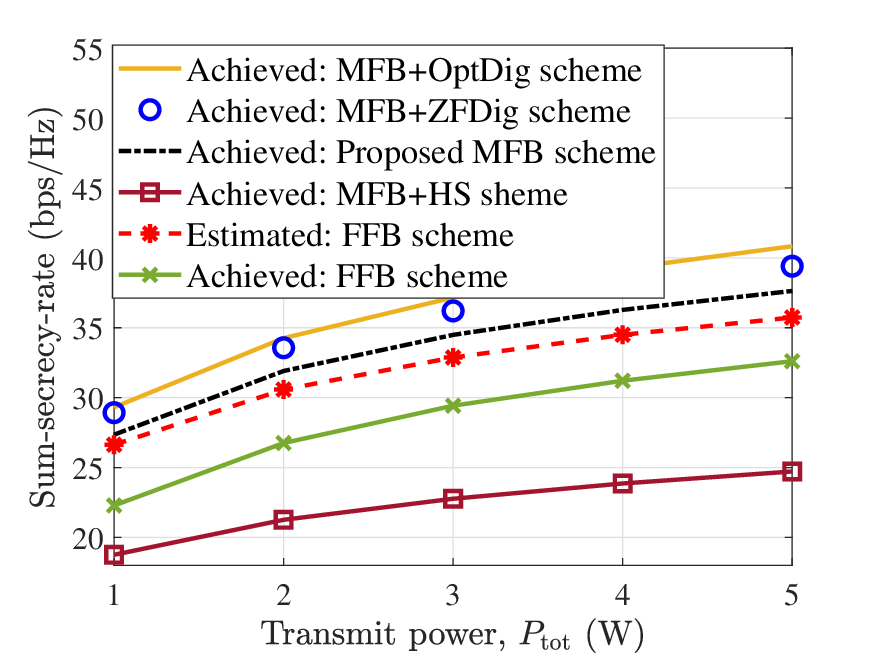}
		\caption{Sum-secrecy-rate v.s. total transmit power.}
		\label{seccrecypower_vs_totpower}
	\end{subfigure}
	\hspace{-1cm}
	\begin{subfigure}[b]{0.49\linewidth}
		\centering
		\includegraphics[width=3in]{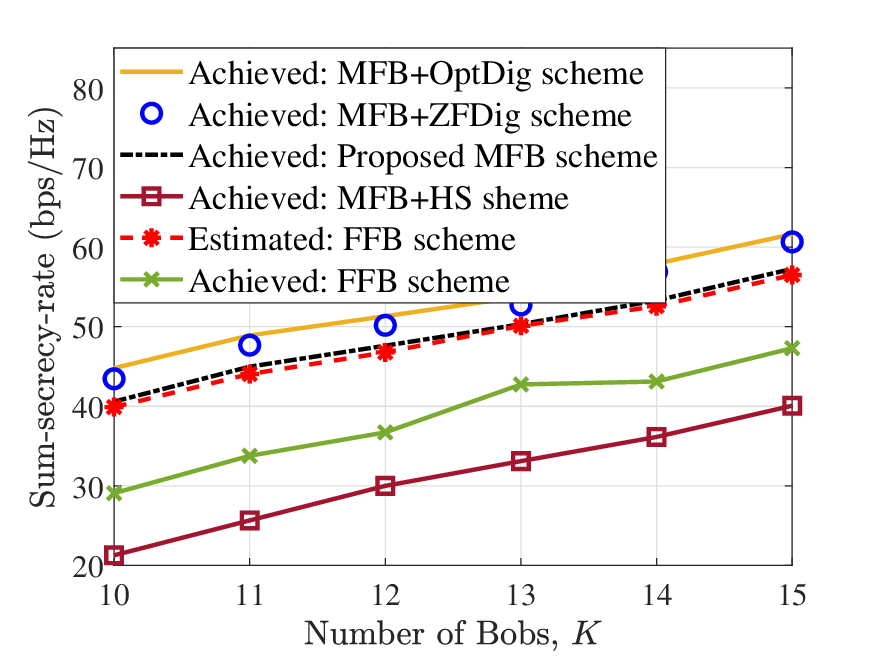}
		\caption{Sum-secrecy-rate v.s. number of Bobs.}
		\label{seccrecypower_vs_diffnumbob}
	\end{subfigure}
	\begin{subfigure}[b]{0.49\linewidth}
		\centering
		\includegraphics[width=3in]{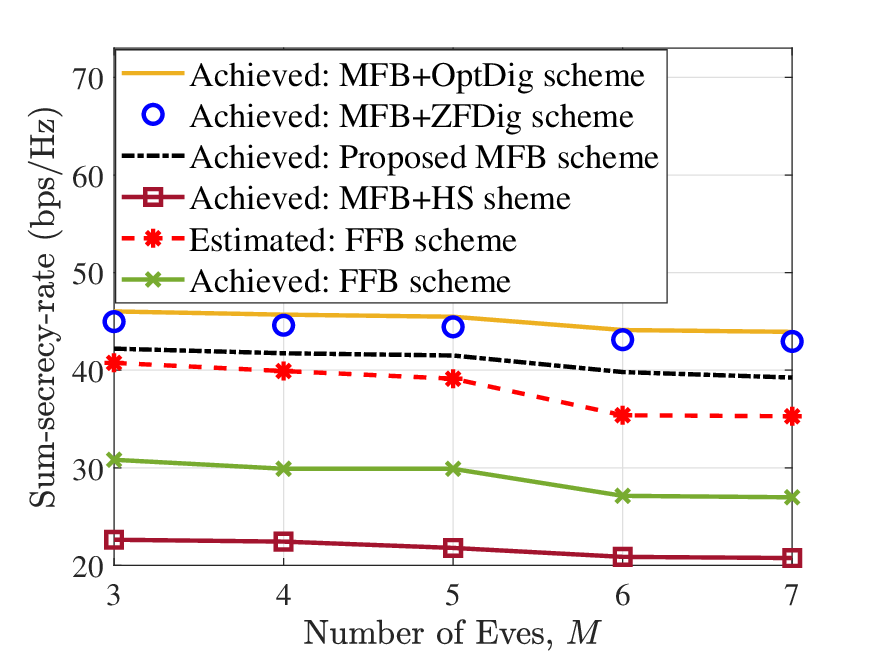}
		\caption{Sum-secrecy-rate  v.s. number of Eves.}
		\label{seccrecypower_vs_diffnumEve}
	\end{subfigure}
	\hspace{-1cm}
	\begin{subfigure}[b]{0.49\linewidth}
		\centering
		\includegraphics[width=3in]{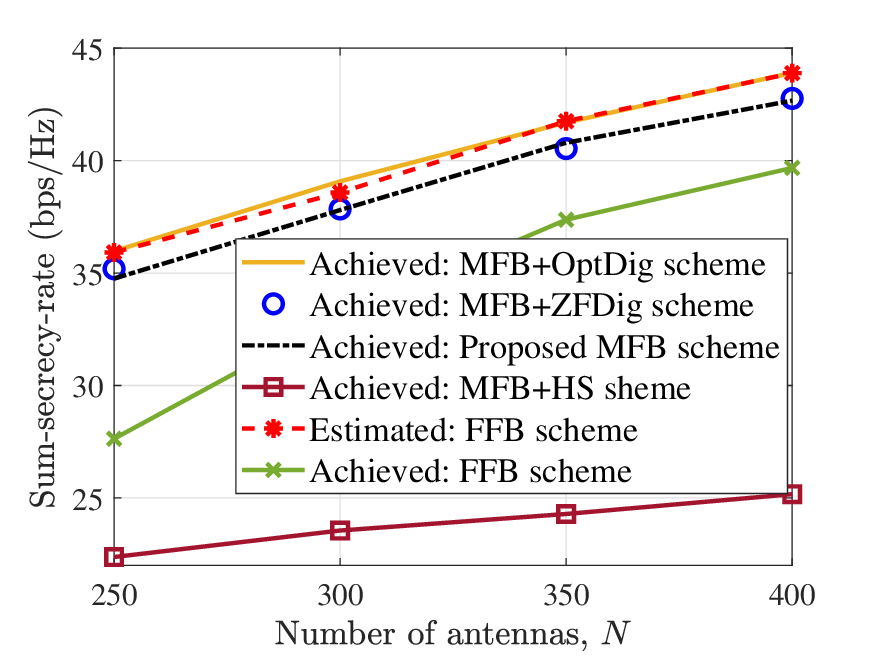}
		\caption{Sum-secrecy-rate  v.s. number of Antennas.}
		\label{seccrecypower_vs_diffnumantenna}
	\end{subfigure}
	\caption{Sum-secrecy-rate v.s. system parameters.}
	\label{secrecy_numerical}
		\vspace{-7pt}
\end{figure*}

\section{Numerical Results}
In this section, we present numerical results to demonstrate the performance gains of our proposed algorithm. We consider a mixed-field PLS system, where a BS with $N=256$ antennas operates at $f=30$ GHz, for which the effective Rayleigh distance of 120 meters at a spatial angle of $\theta=0$ rad. We set $\beta = (\lambda/4\pi)^2 = -62$ dB, $\sigma^2=-80$ dBm, and model all channels based on Rician fading with a Rician factor of 10 dB. The convergence threshold is $\epsilon=10^{-5}$.

For performance comparison, we consider four benchmark schemes as follows
\begin{itemize}
	\item  Far-field based (FFB) scheme, which designs the power allocation for Bobs under the assumption of a purely far-field system.
    \item Mixed-field based scheme with heuristic scheduling (MFB+HS), where all the Bobs satisfying the secure-transmission condition given in \eqref{SecuCond} are scheduled with optimized power allocation.
    \item Mixed-field based scheme + optimal digital beamforming (MFB+OptDig), which optimize the digital beamforming following the algorithm in \textbf{Remark \ref{AlgorithmForhybrid}} under the MRT-based analog beamforming.
    \item Mixed-field based scheme + ZF-based digital beamforming (MFB+ZFDig), which adopts the ZF-based digital beamforming given the MRT-based analog beamforming.
\end{itemize}
     In addition, our proposed algorithm is named as mixed-field based (MFB) scheme. Note that the sum-secrecy-rates calculated from the corresponding algorithms are named as \emph{estimated} rates. For fair comparison, we consider another performance metric, called \emph{achieved} rate, which is obtained under the precise mixed-field system setting.

\subsection{Effect of Total Transmit Power}
In Fig. \ref{secrecy_numerical}(a), we plot the sum-secrecy-rate obtained by different power allocation schemes v.s. total transmit power, where an Eve is located at (0.03 rad, 7 m), and 5 Bobs are located at (0.47 rad, 155 m), (-0.25 rad, 287 m), (0.09 rad, 142 m), (0.03 rad, 208 m), (-0.66 rad, 222 m), respectively.
First, it is observed that the sum-secrecy-rate obtained by all schemes monotonically increases with the total transmit power. 
Second, the proposed MFB scheme only suffers a small performance loss as compared to the MFB+OptDig scheme, especially when $P_{\rm tot}$ is small. Besides, the proposed MFB scheme achieves close
performance with that of the MFB+ZFDig scheme.
Third, it is observed that the achievable sum-secrecy-rate of the FFB scheme is significantly smaller than its estimated value due to channel mismatch, thus necessitating accurate channel modeling to enable efficient power allocation designs tailored to mixed near-field and far-field scenarios.
Besides, the proposed MFB scheme achieves a larger sum-secrecy-rate than the FFB scheme due to the more effective power allocation, which accounts for the impact of mixed-field beam-channel correlation.  
This further highlights the importance of considering practical mixed-field channel model for PLS systems. 
Last, the proposed MFB scheme achieves a significant rate performance gain over the MFB+HS scheme. This is because, in the multi-Bob system, the interferences from other Bobs can be effectively leveraged to mitigate the capability of Eve to intercept signals of target Bobs. As such, the data transmission links that were originally insecure may become secure, thereby leading to enhanced sum-secrecy-rate. 

\subsection{Effect of Number of Bobs} \label{NumericalBob}
Next, we show in Fig. \ref{secrecy_numerical}(b) the effect of the numbers of Bobs on the sum-secrecy-rate obtained by different schemes with $M=3$, where the Eves are randomly located in an area with a radius in $[0.05Z, 0.1Z]$ ($Z$ being the effective Rayleigh distance at a spatial angle of 0 rad) and a spatial angle in $[-\frac{\pi}{3},\frac{\pi}{3}]$, and the Bobs are randomly distributed in an area with a radius in $[Z, 2Z]$ and a spatial angle in $[-\frac{\pi}{3},\frac{\pi}{3}]$. 
First, it is observed that the sum-secrecy-rates achieved by all schemes monotonically increase with the number of Bobs.
This is due to the increasing cumulative interference to Eve from the growing number of Bobs to Eve, which enhances the suppression of  the eavesdropping capability of the Eve and thereby improves the sum-secrecy rate of PLS systems.
Besides, the proposed scheme achieves a substantial performance gain over the MFB+HS scheme. This is because the MFB+HS scheme fails to leverage interferences from Bobs to Eves for mitigating the eavesdropping capability of Eves, resulting in degraded rate performance loss.

\subsection{Effect of Number of Eves}
In Fig. \ref{secrecy_numerical}(c), we evaluate the sum-secrecy-rate performance of the proposed scheme with $K=10$ under different numbers of Eves, where multiple Eves and Bobs are randomly distributed in the same areas as in Subsection \ref{NumericalBob}. First, it is observed that the achieved sum-secrecy-rate obtained by our proposed MFB scheme outperforms that of the benchmark schemes. Second, we observe that the sum-secrecy-rate obtained by all schemes decreases with the number of Eves. This is expected, as for each Bob, the eavesdropping rate generally increases with the number of Eves,  thereby resulting in decreased achievable sum-secrecy-rate. 

\subsection{Effect of Number of Antennas}
Last, in Fig. \ref{secrecy_numerical}(d), we plot the sum-secrecy-rate obtained by different schemes under different numbers of antennas with $M=3$ and $K=13$. First, it is observed that the proposed MFB scheme achieves a larger sum-secrecy-rate than benchmark schemes under different numbers of antennas. 
Note that the achieved sum-secrecy-rate obtained by the proposed MFB scheme is smaller than the estimited one obtained by the FFB scheme. This is intuitive, as a fundamental trade-off exists between two effects in mixed-field systems: As the number of antennas increases, the energy-spread effect becomes more severe, leading to increased information leakage from a greater number of Bobs to Eve, while it also results in increased interference from these Bobs toward Eve, thereby degrading its ability to intercept the target information.
Second, it is observed that the achieved sum-secrecy-rate obtained by the proposed MFB scheme and the FFB scheme increases significantly with the number of antennas. This improvement is attributed to the enhanced array gain for data transmissions to Bobs. In addition, the achieved sum-secrecy-rates obtained by the MFB+HS scheme increases slightly with the number of antennas. This is because, although the array gain for information transmission increases, the energy-spread effect also becomes more severe, which makes more Bobs fail to satisfy the secure-transmission conditions in one-Bob systems (i.e., no power is allocated to these Bobs in this shemes), resulting in limited performance improvement. This further highlights the effectiveness of leveraging interference from multiple Bobs to mitigate eavesdropping capability of Eves, thereby leading to performance rate improvement.

\section{Conclusions}
In this paper, we considered a new mixed-field PLS system where multiple Eves are located in the near-field of the XL-array BS and the Bobs were located in the far-field. An optimization problem was formulated to maximize the sum-secrecy-rate of all Bobs by designing efficient power allocation. To obtain useful insights, we first considered a one-Bob-one-Eve system and characterized the insecure-transmission region of Bob in closed form given the location of Eve. Furthermore, we extended the analysis to a two-Bob-one-Eve system and obtained an approximate solution for power allocation at the BS. Then, for the general case, we proposed an efficient algorithm to obtain a high-quality solution by leveraging the SCA technique. Last, numerical results were presented to verify the effectiveness of the proposed algorithm and the effect of mixed-field system setting on the rate performance.

\appendix
\subsection{Proof for Proposition \ref{SecuAnguRegi}} \label{insecure_thetaB_appendix}
Based on \textbf{Remark~\ref{propertyofcorre}}, given the location of Eve 
$(\theta_{\rm E}, r_{\rm E})$  (and hence a fixed $\beta_2$),  there exist two solutions for $\beta_1$, namely $\beta_1^{-}(\Lambda) < 0$ 
and $\beta_1^{+}(\Lambda) > 0$, such that $G(\beta_1, \beta_2) = \Lambda$. 
According to the secure-transmission condition in \eqref{SecuCond}, the data transmissions to Bob are insecure when $\eta=\left|\mathbf{b}^{H}\left(\theta_{\rm E},r_{\rm E}\right)\mathbf{a}\left(\theta_{\rm B}\right)\right| \geq r_{\rm E}/r_{\rm B}$, which is equivalent to the condition 
$G(\beta_1, \beta_2) \geq \Lambda$ with $\Lambda \triangleq r_{\rm E}/r_{\rm B}$. This condition is satisfied when $\beta_1^{-}(\Lambda) \leq \beta_1 \leq \beta_1^{+}(\Lambda)$.  Then, based on the definition of $\beta_1$ in \textbf{Lemma \ref{correlation}}, by setting $\beta_1 = \beta_1^{+}(\Delta) = \beta_1^{-}(\Delta)$, the insecure-transmission angular region can be obtained,  which is given in \eqref{insecure_thetaB}.
Besides, as the Bob is located in the far-field region, we have $r_{\rm B} > Z$, where $Z$ is the effective Rayleigh distance, thereby completing the proof. 

\subsection{Proof for lemma \ref{securecond_multiuser1}} \label{proof_resecure_condition}
According to the secrecy rate of Bob 1 given in \eqref{TwoBobRs1}, the data transmissions to Bob 1 are secure when
\begin{align}
    &\frac{N P_{{\rm B},1}h^2_{{\rm B},1}}{NP_{{\rm B},2}h^2_{{\rm B},1}\left|\mathbf{a}^H(\theta_{{\rm B},1})\mathbf{a}(\theta_{{\rm B},2})\right|^2+\sigma^2} \nonumber\\&\qquad\qquad> \frac{N P_{{\rm B},1}h^2_{{\rm E}}\left|\mathbf{b}^H(\theta_{{\rm E}},r_{{\rm E}})\mathbf{a}(\theta_{{\rm B},1})\right|^2}{NP_{{\rm B},2}h^2_{{\rm E}}\left|\mathbf{b}^H(\theta_{{\rm E}},r_{{\rm E}})\mathbf{a}(\theta_{{\rm B},2})\right|^2+\sigma^2}. \label{inequation}
\end{align}
By dividing both sides of \eqref{inequation} by $\frac{N P_{{\rm B},1}h^2_{{\rm E}}}{NP_{{\rm B},2}h^2_{{\rm E}}\left|\mathbf{b}^H(\theta_{{\rm E}},r_{{\rm E}})\mathbf{a}(\theta_{{\rm B},2})\right|^2+\sigma^2}$, we obtain
\begin{align}
        &\!\left|\mathbf{b}^H\left(\theta_{\rm E},r_{\rm E}\right)\mathbf{a}\left(\theta_{{\rm B},1}\right)\right|^2\!<\!\frac{r^2_{\rm E}\!+\!\frac{\beta NP_{{\rm B},2}}{\sigma^2}\!\left|\mathbf{b}^H\(\theta_{\rm E},r_{\rm E}\)\mathbf{a}\(\theta_{{\rm B},2}\)\right|^2}{r^2_{{\rm B},1}\!+\!\frac{\beta N P_{{\rm B},2}}{\sigma^2}\!\left|\mathbf{a}^H\(\theta_{{\rm B},1}\)\mathbf{a}\(\theta_{{\rm B},2}\)\right|^2}\!. 
    \end{align}
When $\theta_{{\rm B},1} \neq \theta_{{\rm B},2}$, the Bob 2 $\rightarrow$ Bob 1 interference is negligible, i.e., $\left|\mathbf{a}^H\left(\theta_{{\rm B},1}\right)\mathbf{a}\left(\theta_{{\rm B},2}\right)\right|^2=0$, thus resulting in the secure-transmission condition in \eqref{seccon_for_multiuser}. When $\theta_{{\rm B},1} = \theta_{{\rm B},2}$, we have $\left|\mathbf{a}^H\left(\theta_{{\rm B},1}\right)\mathbf{a}\left(\theta_{{\rm B},2}\right)\right|^2=1$, thus resulting in the secure-transmission condition in \eqref{seccon_for_multiuser}.

\subsection{Proof for Proposition \ref{optimal_poweralloc}} \label{solutioncaseIII}
(1) Proof for Case I:

In this case, only the data transmission to Bob 1 is insecure, meaning that its secrecy rate remains zero regardless of the power allocated to it. As such, all power should be allocated to Bob 2.

(2) Proof for Case II:

When $\left|\!\mathbf{b}^H\!\left(\!\theta_{\rm E}, r_{\rm E}\!\right)\!\mathbf{a}\!\left(\!\theta_{\rm B,2}\!\right)\!\right|\!\rightarrow\!0$ and $\left|\!\mathbf{b}^H\!\left(\!\theta_{\rm E}, r_{\rm E}\!\right)\!\mathbf{a}\!\left(\!\theta_{\rm B,1}\!\right)\!\right|\!\rightarrow\!0$, no data is intercepted by Eve. The optimal solution for power allocation to Problem (P3) follows conventional waterfilling algorithm. 

In high-SNR regimes, when $\left|\!\mathbf{b}^H\!\left(\!\theta_{\rm E}, r_{\rm E}\!\right)\!\mathbf{a}\!\left(\!\theta_{\rm B,2}\!\right)\!\right|$ and $\left|\!\mathbf{b}^H\!\left(\!\theta_{\rm E}, r_{\rm E}\!\right)\!\mathbf{a}\!\left(\!\theta_{\rm B,1}\!\right)\!\right|$ are sufficiently large such that the noise power at Eve can be neglected, the secrecy-rate of the two Bobs can be approximated as
\begin{align}
	R_{\rm S,1} &\!\approx\! \log_2\!\left(\frac{P_{\rm B,1}g^2_{1,1}}{\sigma^2}\right)\!-\!\log_2\!\left(1+\frac{P_{\rm B,1}g^2_{{\rm E},1}}{P_{\rm B,2}g^2_{{\rm E},2}}\right)\!,\label{R1caseIII}\\
	R_{\rm S,2} &\!\approx \!\log_2\!\left(\frac{P_{\rm B,2}g^2_{2,2}}{\sigma^2}\right)\! -\!\log_2\!\left(1+\frac{P_{\rm B,2}g^2_{{\rm E},2}}{P_{\rm B,1}g^2_{{\rm E},1}}\right)\!,\label{R2caseIII}
\end{align}
where $g_{i,j}=\left|\mathbf{h}^H_{\rm B,i} \mathbf{w}_{\rm B,j}\right|$, and $g_{{\rm E},j}=\left|{\mathbf{h}^H_{\rm E}} \mathbf{w}_{\rm B,j}\right|$, $i,j \in \mathcal{K}$.

Then, the sum-secrecy-rate can be rewritten as
\begin{align}
	R_{\rm S,sum}=R_{\rm S,1}+R_{\rm S,2}
	=\log_2\left(A/B\right),
\end{align}
where $A$ and $B$ are given by
\begin{align}
	&A=\left(P^4_{\rm B,1} -2P_{\rm tot}P^3_{\rm B,1} + P^2_{\rm tot}P^2_{\rm B,1} \right)g^2_{1,1}g^2_{2,2}g^2_{{\rm E},1}g^2_{{\rm E},2},\\
	&B=\bigl(P^2_{\rm B,1}\left(g^4_{{\rm E},1}-2g^2_{{\rm E},1}g^2_{{\rm E},2}+g^4_{{\rm E},2}\right) \nonumber\\& \quad +P_{\rm B,1}\left(2P_{\rm tot}g^2_{{\rm E},1}g^2_{{\rm E},2} -2P_{\rm tot}g^4_{{\rm E},2}\right) +P^2_{\rm tot}g^4_{{\rm E},2}\bigr)\sigma^4,
\end{align}
with $P_{\rm B,2}=P_{\rm tot}-P_{\rm B,1}$.

To maximize $R_{\rm S,sum}$, we first obtain the first-order derivative of $\alpha=A/B$ w.r.t. $P_{{\rm B},1}$, given by
\begin{equation}
	\label{eq:xx}
	\frac{\partial \alpha}{\partial P_{{\rm B},1}} = \left(\frac{\partial A}{\partial P_{{\rm B},1}}B - \frac{\partial B}{\partial P_{{\rm B},1}}A\right)/B^2.
\end{equation} 
It is observed that the monotonicity of \eqref{eq:xx} only relies on its numerator, which is factorized into the following form
\begin{align}
	\small
	&f(P_{\rm B,1})=\xi
	P_{\rm B,1}\left(P_{\rm B,1}-P_{\rm tot}\right)\left(P_{\rm B,1}-\frac{P_{\rm tot}g^2_{{\rm E},2}}{g^2_{{\rm E},2}-g^2_{{\rm E},1}}\right)\nonumber\\ &\;\;\;\;\;\times \left(  P_{\rm B,1}-\frac{P_{\rm tot}g_{{\rm E},2}}{g_{{\rm E},2}-g_{{\rm E},1}} \right)\left(P_{\rm B,1}-\frac{P_{\rm tot}g^2_{{\rm E},2}}{g^2_{{\rm E},1}+g^2_{{\rm E},2}}\right),
\end{align}
where $\xi$ is given by $\xi=1/\left(g^2_{{\rm E},2}-g^2_{{\rm E},1}\right)^2$. If $P_{\rm B,1}=P_{\rm tot}g^2_{{\rm E},2}/\left(g^2_{{\rm E},2}-g^2_{{\rm E},1}\right)$ or $P_{\rm B,1}=P_{\rm tot}g_{{\rm E},2}/\left(g_{{\rm E},2}-g_{{\rm E},1}\right)$, the power allocated to Bob 1 falls outside the feasible range $[0,P_{\rm tot}]$, which is not physically reasonable. However, when $P_{\rm B,1}<P_{\rm tot}g^2_{{\rm E},2}/\left(g^2_{{\rm E},1}+g^2_{{\rm E},2}\right)$, the function $\alpha$ is monotonically increasing with $P_{\rm B,1}$, and when $P_{\rm B,1}>P_{\rm tot}g^2_{{\rm E},2}/\left(g^2_{{\rm E},1}+g^2_{{\rm E},2}\right)$, the function $\alpha$ is monotonically decreasing within the region $P_{\rm B,1}\in [0,P_{\rm tot}]$, and $P_{\rm tot}g^2_{{\rm E},2}/\left(g^2_{{\rm E},1}+g^2_{{\rm E},2}\right) \geq 0$ always holds. Therefore, the optimal value that maximize secrecy rate is $P_{\rm B,1}=P_{\rm tot}g^2_{{\rm E},2}/\left(g^2_{{\rm E},1}+g^2_{{\rm E},2}\right)$ and $P_{\rm B,2}=P_{\rm tot}-P_{\rm B,1}$, thus completing the proof.

\bibliographystyle{IEEEtran}
\bibliography{Refs}

\begin{thebibliography}{10}
\providecommand{\url}[1]{#1}
\csname url@samestyle\endcsname
\providecommand{\newblock}{\relax}
\providecommand{\bibinfo}[2]{#2}
\providecommand{\BIBentrySTDinterwordspacing}{\spaceskip=0pt\relax}
\providecommand{\BIBentryALTinterwordstretchfactor}{4}
\providecommand{\BIBentryALTinterwordspacing}{\spaceskip=\fontdimen2\font plus
\BIBentryALTinterwordstretchfactor\fontdimen3\font minus
  \fontdimen4\font\relax}
\providecommand{\BIBforeignlanguage}[2]{{%
\expandafter\ifx\csname l@#1\endcsname\relax
\typeout{** WARNING: IEEEtran.bst: No hyphenation pattern has been}%
\typeout{** loaded for the language `#1'. Using the pattern for}%
\typeout{** the default language instead.}%
\else
\language=\csname l@#1\endcsname
\fi
#2}}
\providecommand{\BIBdecl}{\relax}
\BIBdecl

\bibitem{Tian_Jiachen2023}
J.~Tian, Y.~Han, S.~Jin, and M.~Matthaiou, ``Low-overhead localization and {VR}
  identification for subarray-based {ELAA} systems,'' \emph{IEEE Wireless
  Commun. Lett.}, vol.~12, no.~5, pp. 784--788, Feb. 2023.

\bibitem{Cui_Mingyao2023v10}
M.~Cui, Z.~Wu, Y.~Lu, X.~Wei, and L.~Dai, ``Near-field {MIMO} communications
  for {6G}: Fundamentals, challenges, potentials, and future directions,''
  \emph{IEEE Commun. Mag.}, vol.~61, no.~1, pp. 40--46, Jan. 2023.

\bibitem{Liu_Yuanwei2025}
Y.~Liu, J.~Xu, Z.~Wang, X.~Mu, and L.~Hanzo, ``Near-field communications: What
  will be different?'' \emph{IEEE Wireless Commun.}, vol.~32, no.~2, pp.
  262--270, Apr. 2025.

\bibitem{Zhang_Haiyang2022}
H.~Zhang, N.~Shlezinger, F.~Guidi, D.~Dardari, M.~F. Imani, and Y.~C. Eldar,
  ``Beam focusing for near-field multiuser {MIMO} communications,'' \emph{IEEE
  Trans. Wireless Commun.}, vol.~21, no.~9, pp. 7476--7490, Sep. 2022.

\bibitem{Cui_Mingyao2022}
M.~Cui and L.~Dai, ``Channel estimation for extremely large-scale {MIMO}:
  Far-field or near-field?'' \emph{IEEE Trans. Commun.}, vol.~70, no.~4, pp.
  2663--2677, Apr. 2022.

\bibitem{Lu_Yu2024}
Y.~Lu, Z.~Zhang, and L.~Dai, ``Hierarchical beam training for extremely
  large-scale {MIMO}: From far-field to near-field,'' \emph{IEEE Trans.
  Commun.}, vol.~72, no.~4, pp. 2247--2259, Apr. 2024.

\bibitem{Liu_Linyangside}
L.~Liu, C.~You, Y.~Zhang, and T.~Liu, ``Side angle information assisted
  near-field beam training for {XL}-array communications,'' \emph{IEEE Commun.
  Lett.}, vol.~28, no.~9, pp. 2201--2205, Sep. 2024.

\bibitem{Cui_Mingyao2024}
M.~Cui and L.~Dai, ``Near-field wideband beamforming for extremely large
  antenna arrays,'' \emph{IEEE Trans. Wireless Commun.}, vol.~23, no.~10, pp.
  13\,110--13\,124, Oct. 2024.

\bibitem{YZhang2023}
Y.~Zhang, C.~You, L.~Chen, and B.~Zheng, ``Mixed near- and far-field
  communications for extremely large-scale array: An interference
  perspective,'' \emph{IEEE Commun. Lett.}, vol.~27, no.~9, pp. 2496--2500,
  Sep. 2023.

\bibitem{Zhang_Yunpu2024}
Y.~Zhang and C.~You, ``{SWIPT} in mixed near- and far-field channels: Joint
  beam scheduling and power allocation,'' \emph{IEEE J. Sel. Areas Commun.},
  vol.~42, no.~6, pp. 1583--1597, Jun. 2024.

\bibitem{Li_Yiqing2017}
Y.~Li, M.~Jiang, Q.~Zhang, Q.~Li, and J.~Qin, ``Secure beamforming in downlink
  {MISO} nonorthogonal multiple access systems,'' \emph{IEEE Trans. Veh.
  Technol.}, vol.~66, no.~8, pp. 7563--7567, Aug. 2017.

\bibitem{Zhao_Nan2019}
N.~Zhao, W.~Wang, J.~Wang, Y.~Chen, Y.~Lin, Z.~Ding, and N.~C. Beaulieu,
  ``Joint beamforming and jamming optimization for secure transmission in
  {MISO-NOMA} networks,'' \emph{IEEE Trans. Commun.}, vol.~67, no.~3, pp.
  2294--2305, Mar. 2019.

\bibitem{Zhou_Fuhui2018}
F.~Zhou, Z.~Chu, H.~Sun, R.~Q. Hu, and L.~Hanzo, ``Artificial noise aided
  secure cognitive beamforming for cooperative {MISO-NOMA} using {SWIPT},''
  \emph{IEEE J. Sel. Areas Commun.}, vol.~36, no.~4, pp. 918--931, Apr. 2018.

\bibitem{Cui_Miao2019}
M.~Cui, G.~Zhang, and R.~Zhang, ``Secure wireless communication via intelligent
  reflecting surface,'' \emph{IEEE Wireless Commun. Lett.}, vol.~8, no.~5, pp.
  1410--1414, Oct. 2019.

\bibitem{You_Changsheng2020}
C.~You, B.~Zheng, and R.~Zhang, ``Fast beam training for {IRS}-assisted
  multiuser communications,'' \emph{IEEE Wireless Commun. Lett.}, vol.~9,
  no.~11, pp. 1845--1849, Nov. 2020.

\bibitem{Shao_Xiaodan2022}
X.~Shao, C.~You, W.~Ma, X.~Chen, and R.~Zhang, ``Target sensing with
  intelligent reflecting surface: Architecture and performance,'' \emph{IEEE
  Journal on Selected Areas in Communications}, vol.~40, no.~7, pp. 2070--2084,
  Jul. 2022.

\bibitem{Hu_Guojie2024}
G.~Hu, Q.~Wu, K.~Xu, J.~Si, and N.~Al-Dhahir, ``Secure wireless communication
  via movable-antenna array,'' \emph{IEEE Signal Process Lett.}, vol.~31, pp.
  516--520, 2024.

\bibitem{Zhang_Zheng2024}
Z.~Zhang, Y.~Liu, Z.~Wang, X.~Mu, and J.~Chen, ``Physical layer security in
  near-field communications,'' \emph{IEEE Trans. Veh. Technol.}, vol.~73,
  no.~7, pp. 10\,761--10\,766, Jul. 2024.

\bibitem{Boqun_Zhao2024}
B.~Zhao, C.~Ouyang, X.~Zhang, and Y.~Liu, ``Performance analysis of physical
  layer security: From far-field to near-field,'' \emph{arXiv preprint
  arXiv:2408.10706}, 2024.

\bibitem{Yunpu_Zhang2024_secure}
Y.~Zhang, Y.~Fang, X.~Yu, C.~You, and Y.-J.~A. Zhang, ``Performance analysis
  and low-complexity beamforming design for near-field physical layer
  security,'' \emph{arXiv preprint arXiv:2407.13491}, 2024.

\bibitem{Zhang_Yuchen2024}
Y.~Zhang, H.~Zhang, S.~Xiao, W.~Tang, and Y.~C. Eldar, ``Near-field wideband
  secure communications: An analog beamfocusing approach,'' \emph{IEEE Trans.
  Signal Process.}, vol.~72, pp. 2173--2187, Apr. 2024.

\bibitem{Cui_Mingyao2022_3}
M.~Cui and L.~Dai, ``Channel estimation for extremely large-scale {MIMO}:
  Far-field or near-field?'' \emph{IEEE Trans. Commun.}, vol.~70, no.~4, pp.
  2663--2677, Apr. 2022.

\bibitem{Wu_Xun2024}
X.~Wu, C.~You, J.~Li, and Y.~Zhang, ``Near-field beam training: Joint angle and
  range estimation with {DFT} codebook,'' \emph{IEEE Trans. Wireless Commun.},
  vol.~23, no.~9, pp. 11\,890--11\,903, Sep. 2024.

\bibitem{Wang_Zhaolin2025}
Z.~Wang, X.~Mu, and Y.~Liu, ``Beamfocusing optimization for near-field wideband
  multi-user communications,'' \emph{IEEE Trans. Commun.}, vol.~73, no.~1, pp.
  555--572, Jan. 2025.

\bibitem{Zhang_Yunpu2022}
Y.~Zhang, X.~Wu, and C.~You, ``Fast near-field beam training for extremely
  large-scale array,'' \emph{IEEE Wireless Commun. Lett.}, vol.~11, no.~12, pp.
  2625--2629, Dec. 2022.

\bibitem{Han_Yu2020}
Y.~Han, S.~Jin, C.-K. Wen, and X.~Ma, ``Channel estimation for extremely
  large-scale massive {MIMO} systems,'' \emph{IEEE Wireless Commun. Lett.},
  vol.~9, no.~5, pp. 633--637, May. 2020.

\bibitem{You_Changsheng202038}
C.~You, B.~Zheng, and R.~Zhang, ``Channel estimation and passive beamforming
  for intelligent reflecting surface: Discrete phase shift and progressive
  refinement,'' \emph{IEEE J. Sel. Areas Commun.}, vol.~38, no.~11, pp.
  2604--2620, Nov. 2020.

\bibitem{Mukherjee2012}
A.~Mukherjee and A.~L. Swindlehurst, ``Detecting passive eavesdroppers in the
  {MIMO} wiretap channel,'' in \emph{Proc. IEEE Int. Conf. Acoust., Speech
  Signal Process. (ICASSP)}, Kyoto, Japan, Mar. 2012, pp. 2809--2812.

\bibitem{Xu_Dongfang2020}
D.~Xu, X.~Yu, Y.~Sun, D.~W.~K. Ng, and R.~Schober, ``Resource allocation for
  {IRS}-assisted full-duplex cognitive radio systems,'' \emph{IEEE Trans.
  Commun.}, vol.~68, no.~12, pp. 7376--7394, Dec. 2020.

\bibitem{Anaya_López2023}
G.~J. Anaya-López, J.~P. González-Coma, and F.~J. López-Martínez, ``Leakage
  subspace precoding and scheduling for physical layer security in multi-user
  {XL-MIMO} systems,'' \emph{IEEE Commun. Lett.}, vol.~27, no.~2, pp. 467--471,
  Feb. 2023.

\bibitem{Li_Qiang2013}
Q.~Li and W.-K. Ma, ``Spatially selective artificial-noise aided transmit
  optimization for {MISO} multi-{Eves} secrecy rate maximization,'' \emph{IEEE
  Trans. Signal Process.}, vol.~61, no.~10, pp. 2704--2717, May. 2013.

\bibitem{Ng_Derrick2013}
D.~W.~K. Ng, L.~Xiang, and R.~Schober, ``Multi-objective beamforming for secure
  communication in systems with wireless information and power transfer,'' in
  \emph{Proc. IEEE Annu. Symp. Pers. Indoor Mobile Radio Commun. (PIMRC)},
  London, U.K., Sep. 2013, pp. 7--12.

\bibitem{Jin_Juening2019}
J.~Jin, C.~Xiao, W.~Chen, and Y.~Wu, ``Channel-statistics-based hybrid
  precoding for millimeter-wave {MIMO} systems with dynamic subarrays,''
  \emph{IEEE Trans. Commun.}, vol.~67, no.~6, pp. 3991--4003, Jun. 2019.

\bibitem{Li_Bin2019}
B.~Li, Z.~Fei, C.~Zhou, and Y.~Zhang, ``Physical-layer security in space
  information networks: A survey,'' \emph{IEEE Internet Things J.}, vol.~7,
  no.~1, pp. 33--52, Jan. 2020.

\bibitem{Zhou_Cong2024}
C.~Zhou, C.~Wu, C.~You, J.~Zhou, and S.~Shi, ``Near-field beam training with
  sparse {DFT} codebook,'' \emph{IEEE Trans. Commun.}, pp. 1--1, 2024, Early
  Access.

\bibitem{Anaya_Lopez2022}
G.~J. Anaya-López, J.~P. González-Coma, and F.~J. López-Martínez, ``Spatial
  degrees of freedom for physical layer security in {XL-MIMO},'' in \emph{Proc.
  IEEE 95th Veh. Technol. Conf.}, Helsinki, Finland, Jun. 2022, pp. 1--5.

\bibitem{Sohrabi2016}
F.~Sohrabi and W.~Yu, ``Hybrid digital and analog beamforming design for
  large-scale antenna arrays,'' \emph{IEEE J. Sel. Top. Signal Process.},
  vol.~10, no.~3, pp. 501--513, Apr. 2016.

\end{thebibliography}

\end{document}